\documentclass[12pt]{article}
\usepackage{"conroy_latex_v2"}
\usepackage{array}

\title{\bf  {Combining Clusters for the Approximate Randomization Test}}
\author{Chun Pong Lau\thanks{\mbox{Kenneth C. Griffin Department of Economics, 
The University of Chicago, \href{mailto:ccplau@uchicago.edu}{\texttt{ccplau@uchicago.edu}}.} \\ I am grateful to Alex Torgovitsky, Max Tabord-Meehan, and Azeem Shaikh for their continued guidance and support throughout the project. I also thank St\'{e}phane Bonhomme, Yong Cai, Chris Hansen, Xinran Li, Panos Toulis, and Myungkou Shin for helpful feedback and discussions, as well as participants of the Econometrics Advising Group for valuable comments. All errors are my own.}}

\usepackage{amssymb}
\usepackage{setspace} 
\usepackage{centernot}
\usepackage{marginnote}
\usepackage{multicol}
\usepackage{lscape}
\usepackage{bbding}

\newcommand{\cv}{\widehat{\mathrm{cv}}_n(1-\alpha)}
\newcommand{\og}{\overline{g}}
\newcommand{\oq}{\overline{q}}

\newcommand{\ws}{\widehat{S}}
\newcommand{\wtt}{\widetilde{T}}

\newcommand{\splitatcommas}[1]{%
  \begingroup
  \begingroup\lccode`~=`, \lowercase{\endgroup
    \edef~{\mathchar\the\mathcode`, \penalty0 \noexpand\hspace{0pt plus 1em}}%
  }\mathcode`,="8000 #1%
  \endgroup
}

\usepackage{ragged2e}
\allowdisplaybreaks

\usepackage[sc]{mathpazo}
\begin{document}
\fontfamily{ppl}

\maketitle

\begin{abstract}
This paper develops procedures to combine clusters for the approximate randomization test proposed by \citet*{canayetal2017ecta}. Their test can be used to conduct inference with a small number of clusters and imposes weak requirements on the correlation structure. However, their test requires the target parameter to be identified within each cluster. A leading example where this requirement fails to hold is when a variable has no variation within clusters. For instance, this happens in difference-in-differences designs because the treatment variable equals zero in the control clusters. Under this scenario, combining control and treated clusters can solve the identification problem, and the test remains valid. However, there is an arbitrariness in how the clusters are combined. In this paper, I develop computationally efficient procedures to combine clusters when this identification requirement does not hold. Clusters are combined to maximize local asymptotic power. The simulation study and empirical application show that the procedures to combine clusters perform well in various settings.
\par
\vspace{10pt}
\noindent \textbf{Keywords:} Cluster-robust inference, approximate randomization test, local asymptotic power, difference-in-differences. \par

\end{abstract}

\newpage

\section{Introduction} \label{sec:intro}
It is common to have clustered data in economics where observations within each cluster are correlated with each other. The recent approximate randomization test proposed by \citet*[CRS]{canayetal2017ecta} imposes weak assumptions on the dependence structure and can be used to conduct inference when there is a small number of clusters. However, a main requirement in CRS is that the target parameter has to be identified within each cluster. If this requirement is not satisfied, researchers can combine clusters in a way that ensures the target parameter can be identified within each combined cluster in order to apply CRS. Any method of combining clusters that solves the identification issue leads to a valid test. However, this creates ambiguity on how clusters should be combined. This paper addresses this issue and provides guidance on how to combine clusters.

The issue of not being able to identify the target parameter within each cluster is common in empirical settings. Difference-in-differences (DID) is a leading example of this identification issue within clusters. This is primarily because the treatment indicator equals zero in the control clusters. Some examples include \citet{bloometal2012qje} for firm-level treatment, \citet{burdelinden2013aejae} for village-level treatment, \citet{dinceccokatz2016ej} for country-level treatment, \citet{goodman2016restat} for state-level treatment, and \citet{alsanclaudia2019jpe} for municipality-level treatment. In addition, if each cluster contains only one unit with multiple periods, then the identification issue arises in all clusters when cluster-by-cluster regressions are used. This is because there are more coefficients than observations per cluster when time fixed effects are included. \par

This paper makes two contributions to the literature. First, I analyze the local asymptotic power of CRS. \citet{canayetal2017ecta} showed that their test controls size but did not study local asymptotic power. Analyzing the local asymptotic power of CRS is non-trivial because the number of clusters is fixed and does not go to infinity. Hence, standard large-sample analysis of local power for randomization tests, such as \citet{hoeffding1952ams}, cannot be applied. \par 

Second, I develop computationally efficient procedures to combine clusters by maximizing local asymptotic power. The procedures depend on the chosen significance level. I show that the procedure can be written as a binary linear program when the significance level is ``small.'' For other significance levels, I develop a heuristic to combine clusters. The procedures can be solved quickly and can help eliminate randomness in deciding how to combine the clusters. In the simulation study and empirical exercise, I find that the procedures perform well in various settings.\par

Similar to CRS, this paper is related to the literature on inference with a small number of clusters. \citet{bertrandetal2004qje}  pointed out the issue with standard inference when there is a small number of clusters. Different tests have been proposed in this literature apart from CRS. Some recent methods include \citet{besteretal2011joe}, \citet{ibraimovmuller2010jbes, ibraimovmuller2016restat}, and \citet{hagemann2022wp}. The Wild cluster bootstrap popularized by \citet{cameronetal2008restat} has been shown to be valid under strong homogeneity conditions when the number of clusters is fixed by \citet{canaysantosshaikh2021restat}. The Wild cluster bootstrap has also been studied in other settings, such as \citet{mackinnonwebb2023jae}, \citet{djogbenouetal2019joe}, and \citet{webb2023cje}. \citet{caietal2023jem} provide a users-guide for CRS. \citet{cai2024wp} studies finite sample properties of the sign test. Recently, \citet{caoetal2022wp} proposed a data-driven procedure to partition observations into clusters and conduct valid inference. Their procedure focuses on spatially indexed data and assumes there exists a dissimilarity measure between observations. This paper is different in that I take the clusters as given and focus on combining the clusters when the target parameter cannot be identified within each cluster. This paper follows the above papers and uses the model-based approach in clustering. See \citet{abadieetal2022qje} for a design-based analysis. For surveys on the literature of conducting inference with a fixed number of clusters, see, for instance, \citet{cameronmiller2015jhr}, \citet{conleyetal2018jar}, and \citet{mackinnonetal2023joe}. 

The remainder of this paper is organized as follows. Section \ref{sec:2} reviews the approximate randomization test proposed by CRS and outlines two motivating examples for combining clusters. Section \ref{sec:3} analyzes the local asymptotic power of CRS in order to develop the necessary tools for combining clusters. Section \ref{sec:4} presents the data-driven procedures to combine clusters. Sections \ref{sec:5} and \ref{sec:6} present results from Monte Carlo simulations and an empirical application, respectively. Section \ref{sec:7} concludes. All proofs can be found in the appendix.

\section{The test and motivating examples} \label{sec:2}
\subsection{Setup} \label{sec:2.1}
Consider the following linear regression model of clustered data. Clusters are indexed by $j \in \cJ \equiv \{1, \ldots, q\}$ and units in the $j$th cluster are indexed by $i \in \cI_{n,j} \equiv \{1, \ldots, n_j\}$:
\begin{equation}
	\label{eq:1}
	Y_{i,j} = X_{i,j}'\beta + U_{i,j},
\end{equation}
where $Y_{i,j}$ is an outcome variable, $X_{i,j} \in \bR^{d_x}$ is a vector of covariates,  $U_{i,j} \in \bR$ is a residual, and $\beta \in \bR^{d_x}$ is an unknown parameter. The total number of observations is $n \equiv \sum_{j \in \cJ} n_j$. \par 

The goal is to test the hypothesis:
\begin{equation}
	\label{eq:2-hypo}
	H_0: c'\beta = \lambda
	\quad \text{vs.} \quad 
	H_1: c'\beta \neq \lambda,
\end{equation}
for a given vector $c \in \bR^{d_x} \backslash \{0_{d_x}\}$ (where $0_{d_x}$ is a vector of $d_x$ zeros) and $\lambda \in \bR$ at level $\alpha \in (0, 1)$.

\subsection{Review of the approximate randomization test} \label{sec:2.2}
This section briefly reviews the approximate randomization test by CRS in the context of the model in Section \ref{sec:2.1}. The following five-step procedure is taken from \citet{caietal2023jem}.

\begin{algo}[Algorithm 2.1 from \citet{caietal2023jem}] \label{algo:2.1} \text{ }
\begin{description}
	\item[Step 1:] For each $j \in \cJ$, regress $Y_{i,j}$ on $X_{i,j}$ using only the observations from cluster $j$. Denote $\widehat \beta_{n,j}$ as the corresponding estimator of $\beta$ for each $j \in \cJ$.
	\item[Step 2:] For each $j \in \cJ$, construct the test statistic
	\begin{equation}
		\label{eq:3-test}
		T_n \equiv \left| \frac{1}{q} \sum^q_{j=1} \widehat S_{n,j} \right|,
	\end{equation}
	where $\ws_{n,j} \equiv \sqrt{n_j}(c'\widehat\beta_{n,j} - \lambda)$.
	\item[Step 3:] Define $\bG \equiv \{1, -1\}^q$. For each $g \equiv (g_1, \ldots, g_q) \in \bG$, define
	\begin{equation}
		\label{eq:4-rtest}
		T_n(g) \equiv \left| \frac{1}{q} \sum^q_{j=1} g_j \widehat S_{n,j} \right|.
	\end{equation}
	\item[Step 4:] Compute the $(1-\alpha)$-quantile of $\{T_n(g) : g \in \bG\}$ as 
	\[
		\widehat{\text{cv}}_n(1-\alpha)
		\equiv \inf \left\{ u \in \bR : 
			\frac{1}{|\bG|} \sum_{g \in \bG} \ind[T_n(g) \leq u] \geq 1-\alpha 
			\right\}.
	\]
	\item[Step 5:] Compute the test as 
	\begin{equation}
		\label{eq:test-1}
		\phi_n \equiv \ind[T_n > \widehat{\text{cv}}_n(1-\alpha)].
	\end{equation}
\end{description}
\end{algo}

CRS requires the clusters to satisfy a joint convergence in distribution requirement and the corresponding limiting random variables to be invariant to sign changes. With these weak assumptions, the test is valid for a fixed number of clusters as $n \longrightarrow \infty$ under the null hypothesis. See Section 3 in \citet{canayetal2017ecta} and Section 4 in \citet{caietal2023jem} for more details. The assumptions are stated as follows. \par

\begin{assu} \label{assu:2.2} 
Let $\{\widehat \beta_{n,j} : j \in \cJ\}$ be the estimators of $\beta$ obtained from cluster-by-cluster regressions in Step 1 of Algorithm \ref{algo:2.1}. Assume that:
	\begin{enumerate}
		\item \label{assu:2.2.2} $\frac{\sqrt{n_j}}{\sqrt{n}} \longrightarrow \xi_j > 0$ for any $j \in \cJ$.
		\item \label{assu:2.2.1} 
		$\begin{pmatrix}
			\sqrt{n_1} (\widehat\beta_{n,1} - \beta) \\
			\vdots	\\
			\sqrt{n_q} (\widehat\beta_{n,q} - \beta) \\
		\end{pmatrix}
		\Dt
		\begin{pmatrix}
			S_1 \\ 
			\vdots \\
			S_q
		\end{pmatrix}$,
		where $S_j \sim N(0, \Sigma_j)$ for each $j \in \cJ$ and $\Sigma_j$ is positive definite.
		\item \label{assu:2.2.3} $S_j \indep S_k$ for any $j, k \in \cJ$ with $j \neq k$.
	\end{enumerate}
\end{assu}
Assumption \ref{assu:2.2}.\ref{assu:2.2.2} requires that none of the clusters have a negligible size. Assumption \ref{assu:2.2}.\ref{assu:2.2.1} requires that $\widehat\beta_{n,j}$ obtained from cluster-by-cluster regression is asymptotically normal. It is satisfied when there is weak dependence within clusters so that an appropriate central limit theorem applies. Note that Assumption 3.1 of  \citet{canayetal2017ecta} imposes a weaker version in that it requires the limiting random variable to be symmetric, not necessarily normal. Assumption \ref{assu:2.2}.\ref{assu:2.2.3} is a standard assumption that requires the clusters to be independent of each other.

Assumption \ref{assu:2.2} with $\cJ$ being the set of clusters is appropriate when the target parameter can be identified within each cluster. It is not immediately satisfied when there is an identification problem within clusters. However, suppose the clusters can be combined into a coarser level $\omega$ so that the target parameter is identified in each combined cluster. Then, Assumption \ref{assu:2.2} holds when $\cJ$ is replaced by $\omega$. In addition, CRS controls size when $\cJ$ is replaced by $\omega$.

\begin{re}
The test statistic defined in \eqref{eq:3-test} is unstudentized. When the target parameter in the hypothesis is a scalar, studentizing the test statistic or not does not affect the results of the test. For details, see Section 3 of \citet{caietal2023jem}.
\end{re}

\begin{re}
CRS also offers a slightly different version of the test that is randomized to ensure that the rejection probability equals the significance level $\alpha$ under the null. Algorithm \ref{algo:2.1} uses a nonrandomized version of the test so that there is no randomness in reporting the result. See \citet{canayetal2017ecta} for simulations comparing the randomized and nonrandomized versions of the test. They showed that using the nonrandomized version leads to a rejection probability slightly less than $\alpha$.
\end{re}

\subsection{Motivating examples: the need to combine clusters}

Recall that Step 1 of Algorithm \ref{algo:2.1} requires the researcher to obtain $\widehat \beta_{n,j}$ using cluster-by-cluster regressions. This is not possible when the target parameter cannot be identified within each cluster. This section outlines two common empirical settings with this identification issue.

\begin{eg}[Clustered regression] \label{eg:3.1}
	Consider the following simplified version of the linear model \eqref{eq:1}:
	\[
		Y_{i,j} = \beta_0 + \beta_1 D_j + U_{i,j},
	\]
	where $D_j \in \{0, 1\}$ is a cluster-level treatment indicator and $\beta_0, \beta_1, U_{i,j} \in \bR$. Let $\cJ = \cJ_0 \cup \cJ_1$ where $\cJ_0$ is the set of control clusters and $\cJ_1$ is the set of treated clusters, so that $D_j = \ind[j \in \cJ_1]$ for all $j \in \cJ$. While it is possible to estimate $(\beta_0, \beta_1)$ using all observations, it is not possible to do so using cluster-by-cluster regressions due to collinearity. \par 
	As suggested by CRS, researchers need to combine treated and control clusters in order to apply Algorithm \ref{algo:2.1}. Note that any combination of the clusters does not affect the validity of CRS because the test relies on $n \longrightarrow \infty$, while holding the number of clusters $q$ fixed. Under the scenario where $|\cJ_1| = |\cJ_0|$, one way is to pair one treated cluster with one control cluster in this example.
\end{eg}

\begin{eg}[Difference-in-differences] \label{eg:3.2}
Consider the following static two-way fixed effects specification of the DID model:
\[
	Y_{j,t} = \alpha_j + \phi_t + D_{j,t} \beta + U_{j,t},
\]
where clusters are indexed by $j \in \cJ \equiv \{1, \ldots, q\}$, time is indexed by $t \in \cT \equiv \{1, \ldots, T\}$, $\alpha_j$ is the cluster fixed effect, $\phi_t$ is the time fixed effect, and $D_{j,t} \in \{0, 1\}$ is the treatment indicator. Similar to the last example, let $\cJ = \cJ_0 \cup \cJ_1$ where $D_{j,t} = 1$ for $j \in \cJ_1$ and $t > t_0$ for some $1 < t_0 < T$ and $D_{j,t} = 0$ otherwise. \par

Although $D_{j,t}$ has variation across $t \in \cT$ within clusters for $j \in \cJ_1$, the other clusters are such that $D_{j,t} = 0$ for all $j \in \cJ_0$ and $t \in \cT$. In addition, it is not possible to identify all parameters using cluster-by-cluster regression because there are more parameters than observations when the cluster and time fixed effects are included. Nevertheless, the identification problem can be solved by combining treated and control clusters.
\end{eg}

\section{Local asymptotic power analysis} \label{sec:3}

This section analyzes the local asymptotic power of CRS in order to develop the tools for combining the clusters in Section \ref{sec:4}.

\subsection{Notations} \label{sec:3.1}

Let $\delta \in \bR$ be a local parameter such that 
\begin{equation}
	\label{eq:3.1}
	\lambda = c'\beta + \frac{\delta}{\sqrt{n}}.
\end{equation}
The object of interest is the limiting rejection probability of CRS under the local alternative as described in \eqref{eq:3.1}, i.e.,
\begin{equation}
	\label{eq:3.2}
	\pi(\delta, \alpha) \equiv \lim_{n \to \infty} \bP_\delta[T_n > \widehat{\text{cv}}_n(1-\alpha)],
\end{equation}
where the critical value $\widehat{\text{cv}}_n(1-\alpha)$ is the $(1-\alpha)$-quantile of $\{T_n(g) : g \in \bG\}$ as defined in Step 4 of Algorithm \ref{algo:2.1}, and $\bP_\delta$ is used to emphasize the dependence of the distribution of the data on the local parameter $\delta$.

As mentioned in the introduction, the asymptotic framework here is $n \longrightarrow \infty$, while holding the number of clusters $q$ fixed. With $q$ being fixed, $|\bG|$ is also fixed. Hence, the ``randomization distribution'' may not settle down even when $n \longrightarrow \infty$. Therefore, results on the large-sample properties of the local power for randomization tests such as \citet{hoeffding1952ams} cannot be applied here. See Remark 3.5 of \citet{canayetal2017ecta} for more discussion.

In order to compute the local asymptotic power of CRS, note that there can be at most $2^{q-1}$ unique values in $\{T_n(g) : g \in \bG\}$. This follows by the definition of $T_n(g)$ in equation \eqref{eq:4-rtest} that
\begin{equation}
	\label{eq:sign-dup}
	T_n(g) 
	=
	\left| \frac{1}{q} \sum^q_{j=1} g_j \widehat S_{n,j} \right|
	=
	\left| \frac{1}{q} \sum^q_{j=1} (-g_j) \widehat S_{n,j} \right|
	= T_n(-g),
\end{equation}
for any $g \in \bG$. Therefore, I focus on the set of sign changes $\bG_U \equiv \{g \in \bG : g_1 = 1\}$ that give the unique values of $T_n(g)$ to avoid duplications as in \eqref{eq:sign-dup}. By construction, $|\bG_U| = 2^{q-1}$. The critical value $\widehat{\text{cv}}_n(1-\alpha)$ can also be evaluated as the $(1 - \alpha)$-quantile of $\{T_n(g) : g \in \bG_U\}$ instead of $\{T_n(g) : g \in \bG\}$. \par 

For notational convenience, define $\bG_{U, -1} \equiv \bG_U \backslash \{1_q\}$, where $1_q \equiv (1, \ldots, 1)$ is the identity transformation. In addition, let $L \equiv |\bG_{U, -1}| = 2^{q-1} - 1$ and $\{\og_1, \ldots, \og_L\}$ be the set of all distinct sign changes in $\bG_{U, -1}$.

\subsection{Local asymptotic power} \label{sec:3.2}
Following the discussion in the last section, the event $\{T_n > \cv\}$ can be interpreted in terms of the sign changes in $\bG_U$. In particular, this event collects all arrangements of $\{T_n(g) : g \in \bG_U\}$ such that $T_n$ is one of the largest $K \equiv \lfloor \alpha |\bG_U|\rfloor$ terms. But recall in \eqref{eq:4-rtest} that each $T_n(g)$ is defined as first summing the random variables $\ws_{n,j}$ over $j \in \cJ$ multiplied by the sign changes, and then taking the absolute value. Therefore, computing $\{T_n > \cv\}$ is a non-trivial task because all the terms in $\{T_n(g) : g \in \bG_U\} $ are dependent on each other.  \par

Observe that computing the probability that $T_n$ is the $k$th largest term in $\{T_n(g) : g \in \bG_U\}$ amounts to comparing $T_n$ with $T_n(g)$ for each $g \in \bG_{U, -1}$. In particular, each comparison involves checking whether $T_n > T_n(g)$ or not for all $g \in \bG_{U, -1}$. The following lemma shows that these comparisons can be expressed explicitly in terms of the partial sums of $\ws_{n,j}$ depending on the sign change $g \in \bG_{U, -1}$.

\begin{lem} \label{lem:3.3.3}
For any $g, h \in \bG$, let 
	\[
		\cJ_{\mathrm{same}}(g, h) \equiv \{j \in \cJ : g_j = h_j\}
	\]
	be the set of indices such that the corresponding components of sign changes in $g$ and $h$ are the same and 
	\[
		\cJ_{\mathrm{diff}}(g, h) \equiv \{j \in \cJ: g_j \neq h_j\}
		=  \{j \in \cJ: g_j =- h_j\}
	\]
	be the set of indices such that the corresponding components of sign changes in $g$ and $h$ are different.  Then,
	\begin{align}
		\label{eq:lem:3.1}
		\begin{split}
		\{T_n(h) > T_n(g)\}
		& = 
		\left\{ \sum_{j \in \cJ_{\mathrm{same}}(g, h)} h_j \ws_{n,j} > 0, \sum_{j \in \cJ_{\mathrm{diff}}(g, h)} h_j \ws_{n,j} > 0 \right\} 	\\
		& \quad\quad \bigcup \left\{ \sum_{j \in \cJ_{\mathrm{same}}(g, h)} h_j \ws_{n,j} < 0, \sum_{j \in \cJ_{\mathrm{diff}}(g, h)} h_j \ws_{n,j} < 0 \right\} .
		\end{split}
	\end{align}
\end{lem}

The above lemma shows that each comparison of $T_n(h) > T_n(g)$ can be written as two disjoint events. The lemma helps to simplify the computation of local asymptotic power. \par 

Next, recall from the beginning of this subsection that computing the local asymptotic power requires computing the probability that $T_n$ is the $k$th largest term in $\{T_n(g) : g \in \bG_U\}$ for each $k = 1, \ldots, \lfloor \alpha |\bG_U| \rfloor$. The next lemma considers the probability that $T_n$ is less than a subset of terms from $\{T_n(g) : g \in \bG_{U, -1}\}$ and is larger than the remaining terms. It uses Lemma \ref{lem:3.3.3} to represent this as a comparison using partial sums. 

\begin{lem} \label{lem:3.3}
Let Assumption \ref{assu:2.2} hold. Let $\delta \in \bR$ be a local alternative parameter as in \eqref{eq:3.1}. In addition, define the following notations.
\begin{itemize}
\item For each $g \in \bG_{U, -1}$, define the partial sums
\begin{align*}
	V_{\mathrm{same}}(g, \delta)
	& \equiv
	\sum_{j \in \cJ_{\mathrm{same}}(g, 1_q)} (Z_j + \xi_j \delta),		\\
	V_{\mathrm{diff}}(g, \delta)
	& \equiv
	\sum_{j \in \cJ_{\mathrm{diff}}(g, 1_q)} (Z_j + \xi_j \delta),
\end{align*}
where $Z_j \equiv c'S_j$ for each $j \in \cJ$. See Assumption \ref{assu:2.2} for the definitions of $\xi_j$ and $S_j$.  
\item For each $g \in \bG_{U, -1}$ and any $\cH \subseteq \bG_{U, -1}$, define
\[
	\cF^{(\ell)}(g, \cH, \delta)
	\equiv 
	\begin{cases}
		\displaystyle
		\left\{ 
			\kappa_\ell V_{\mathrm{same}}(g, \delta) >  0 , \  
			\kappa_\ell  V_{\mathrm{diff}}(g, \delta) < 0
		\right\}
		& , \ \ g  \in \cH \\
		\displaystyle
		\left\{
			\kappa_\ell  V_{\mathrm{same}}(g, \delta) > 0 , \
			\kappa_\ell V_{\mathrm{diff}}(g, \delta) > 0
		\right\}
		& , \ \ g \notin \cH
	\end{cases},
\]
for $\ell \in \{1, 2\}$ with $\kappa_1 = 1$ and $\kappa_2 = -1$.
\item Define $\bM \equiv \{1, 2\}^L$ and write each $m \in \bM$ as $(m_1, \ldots, m_L)$.
\end{itemize}
Let $\cH_1$ be a subset of $\bG_{U, -1}$ of size $k \neq L$ and $\cH_2 \equiv \bG_{U, -1} \backslash \cH_1$. Then,
\[
	\lim_{n\to\infty}
	\bP_\delta[\text{$T_n(h_1) > T_n \  \forall h_1 \in  \cH_1 \ \mathrm{and} \ T_n > T_n(h_2)\  \forall h_2 \in \cH_2$}]
	= 
	\sum_{m \in \bM}
	\bP\left[ \bigcap^{L}_{l=1} \cF^{(m_l)}(\og_l, \cH_1, \delta) \right].
\]
\end{lem}
In the above lemma, the partial sums $V_{\text{same}}(g, \delta)$ and $V_{\text{diff}}(g, \delta)$ are based on the partial sums in \eqref{eq:lem:3.1}. 
They are similar to those in Lemma \ref{lem:3.3.3} but taking $n \longrightarrow \infty$, and are evaluated under the local alternative as in \eqref{eq:3.1}.
Since the comparisons are about $T_n$ and $T_n(g)$ for $g \in \bG_{U, -1}$, the lemma sets $h = 1_q$ directly from Lemma \ref{lem:3.3.3}. They describe the event that $T_n$ is larger than or smaller than another $T_n(\og_l)$ for $l = 1, \ldots, L$.  \par 

Next, the following is a mild assumption used to control the ties of $\{T_n(g) : g \in \bG_U\}$.

\begin{assu}
	\label{assu:3.3.4}
	Let $\bW \equiv \{ w = (w_1, \ldots, w_q) \in \bR^q : w_j \neq 0 \text{ for at least one $1 \leq j \leq q$}\}$. For any $w \in \bW$ and $w_0 \in \bR$, 
	\[
		w_0 + \sum^q_{j=1} w_j Z_j \neq 0,
	\]
	with probability 1.
\end{assu}

See Lemmas S.5.1 to S.5.3 in \citet{canayetal2017ecta-supp} on how Assumption \ref{assu:3.3.4} is related to various commonly-used test statistics. As they have stated in their lemmas, the assumption is satisfied when the distribution of $Z_j$ is absolutely continuous with respect to the Lebesgue measure for each $j \in \cJ$. \par 

The following theorem states the main result that expresses the local asymptotic power using the above two lemmas. The main idea of the theorem is to consider different ordering of the terms in $\{T_n(g) : g \in \bG_U\}$ and keep those such that $T_n$ is one of the $\lfloor \alpha |\bG_U|\rfloor$ largest terms.

\begin{thm} \label{prop:3.3.4} 
Consider the problem of testing \eqref{eq:2-hypo}. Let Assumptions \ref{assu:2.2} and \ref{assu:3.3.4} hold, $\alpha \in (0, 1)$, $K \equiv \lfloor \alpha |\bG_U|\rfloor$, and $\delta \in \bR$ be a local alternative parameter as in \eqref{eq:3.1}. In addition, let $\bH_k$ be the collection of all distinct size $k$ subsets of $\bG_{U, -1}$. Then, the local asymptotic power can be written as
\begin{equation}
	\label{eq:prop3.2-eq}
 	\pi(\delta, \alpha)
	= 
	\sum^K_{k=1} 
	\sum_{\cH \in \bH_{k-1}}
	\sum_{m \in \bM}
	\bP\left[ \bigcap^{L}_{l=1} \cF^{(m_l)}(\og_l, \cH, \delta) \right],
\end{equation}
where $\bM$ and $\cF^{(m_l)}(\overline g_l, \cH, \delta)$ are the same as defined in Lemma \ref{lem:3.3}.
\end{thm}

In words, equation \eqref{eq:prop3.2-eq} considers the ordering of the terms in $\{T_n(g) : g \in \bG_U\}$ and keeps those such that $T_n$ is one of the $\lfloor \alpha |\bG_U|\rfloor$ largest terms. Hence, there are three summations in the local asymptotic power expression. The first summation (over $k$) sums over the events that $T_n$ is the $k$th largest term. For each $k$, the second summation (over $\cH$) sums over different combinations of the events where $T_n(g) > T_n$ for all $g \in \cH \subseteq \bG_{U, -1}$ with $|\cH| = k-1$. The third summation (over $m$) considers different combinations of the events $\{\cF^{(m_l)}(\og_l, \cH, \delta)\}_{l =1}^L$  as in Lemma \ref{lem:3.3}. The expression in Theorem \ref{prop:3.3.4} involves computing the joint probability of multivariate normal distributions. Closed-form expressions may not be always available, but the expression can be computed numerically. \par 

When $K = 1$, the expression in \eqref{eq:prop3.2-eq} simplifies greatly. The reason is that with $K = 1$, $\bH_0 = \emptyset$, and the corresponding event $\{T_n > \widehat{\mathrm{cv}}_n(1 - \alpha)\}$ means that $T_n > T_n(g)$ for all $g \in \bG_{U, -1}$. The following corollary shows that when $K = 1$, the local asymptotic power in Theorem \ref{prop:3.3.4} simplifies to a sum of two products of univariate normal cumulative distribution functions. Since $|\bG_{U}| = 2^{q-1}$, the range of significance levels $\alpha$ such that $\lfloor \alpha |\bG_U|\rfloor = 1$ is $\alpha  \in [ \frac{1}{2^{q-1}} , \frac{1}{2^{q-2}})$.

\begin{cor} \label{cor:3.3.5} 
Consider the problem of testing \eqref{eq:2-hypo}. Let Assumptions \ref{assu:2.2} and \ref{assu:3.3.4} hold and $\delta \in \bR$ be a local alternative parameter as in \eqref{eq:3.1}. For any $\alpha  \in [ \frac{1}{2^{q-1}} , \frac{1}{2^{q-2}})$, the local asymptotic power can be written as
\[
		 \pi(\delta, \alpha)
		 = \pi_{\mathrm{L}}(\delta) + \pi_{\mathrm{R}}(\delta),
\]
where $\sigma_j^2 \equiv c'\Sigma_j c$ for each $j \in \cJ$, 
	\begin{align*}
		 \pi_{\mathrm{L}}(\delta) 
		 &\equiv
		 \prod_{j \in \cJ} \Phi \left( -\frac{\xi_j \delta}{\sigma_j} \right),	\\
		 \pi_{\mathrm{R}}(\delta) 
		 &\equiv
		  \prod_{j \in \cJ} \left[1 - \Phi \left( -\frac{\xi_j \delta}{\sigma_j} \right)\right].
	\end{align*}
\end{cor}

A few useful properties of the local asymptotic power for $K=1$ above are summarized below.

\begin{ppt} \label{ppt:3.7}
Consider the problem of testing \eqref{eq:2-hypo}. Let Assumptions \ref{assu:2.2} and \ref{assu:3.3.4} hold, $\delta \in \bR$ be a local alternative parameter as in \eqref{eq:3.1}, and $\alpha \in [ \frac{1}{2^{q-1}} , \frac{1}{2^{q-2}})$. The following properties hold.
	\begin{enumerate}
		\item \label{ppt:3.7.1} $\frac{\partial  \pi_{\mathrm{L}}(\delta) }{\partial \delta} < 0$ and  $\frac{\partial  \pi_{\mathrm{R}}(\delta) }{\partial \delta} > 0$.
		\item \label{ppt:3.7.2} $\pi_{\mathrm{L}}(0) = \pi_{\mathrm{R}}(0) = \frac{1}{2^q}$.
		\item \label{ppt:3.7.3} 
		\begin{enumerate}
			\item If $\delta < 0$, $\pi_{\mathrm{L}}(\delta) > \frac{1}{2^q} > \pi_{\mathrm{R}}(\delta)$. 
			\item If $\delta > 0$, $\pi_{\mathrm{R}}(\delta) > \frac{1}{2^q} > \pi_{\mathrm{L}}(\delta)$.
		\end{enumerate}
	\end{enumerate}
\end{ppt}

Property \ref{ppt:3.7}.\ref{ppt:3.7.1} shows that $\pi_{\mathrm{L}}(\delta)$ and $\pi_{\mathrm{R}}(\delta)$ move in opposite directions in $\delta$. Property \ref{ppt:3.7}.\ref{ppt:3.7.2} shows that $\pi_{\mathrm{L}}(\delta)$ and $\pi_{\mathrm{R}}(\delta)$ intersect once at $\delta = 0$. Property \ref{ppt:3.7}.\ref{ppt:3.7.3} is implied by the first two properties. It states that when $\delta \neq 0$, then one of $\pi_{\mathrm{L}}(\delta)$ or $\pi_{\mathrm{R}}(\delta)$ is larger than the other. The smaller term is bounded above by $\frac{1}{2^q}$, which is usually a small value.  Finally, the functions $\pi_{\text{L}}(\delta)$ and $\pi_{\text{R}}(\delta)$ can be interpreted as the local asymptotic power of one-sided tests.

\begin{eg} \label{eg:3.6} Suppose that $q = 5$, and $\frac{\xi_j}{\sigma_j} = 1$ for $j = 1, \ldots, 5$. Figure \ref{fig:power-demo} plots $\pi_{\text{L}}(\delta)$ and $\pi_{\text{R}}(\delta)$ for $\delta \in [-2, 2]$ and demonstrates Property \ref{ppt:3.7}.  The figure shows that $\pi_{\text{L}}(\delta)$ is decreasing in $\delta$ and $\pi_{\text{R}}(\delta)$ is increasing in $\delta$. They intersect once at $\delta = 0$. In addition, $\pi_{\text{L}}(\delta)$ is larger than $\pi_{\text{R}}(\delta)$ when $\delta < 0$, and vice versa.

\begin{figure}[!ht]
	\centering
	\caption{Illustration of the local asymptotic power function with $q = 5$.}
	\label{fig:power-demo}
	\includegraphics[scale=1]{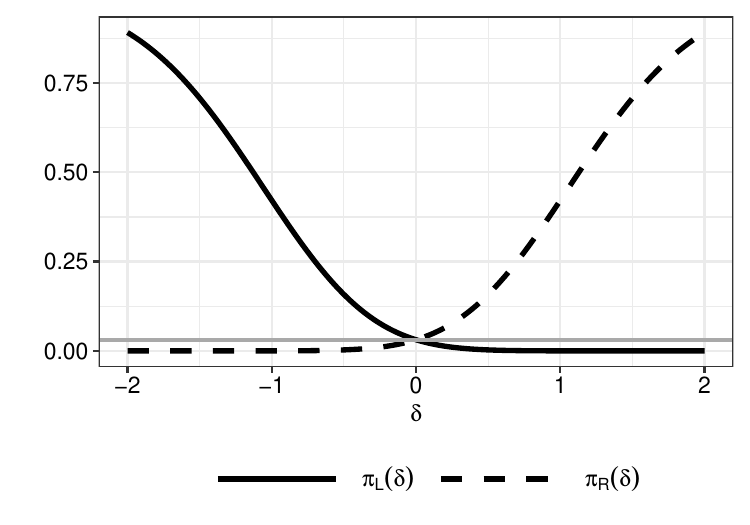}
\end{figure}

\end{eg}

\section{Combining clusters} \label{sec:4}

This section develops data-driven procedures to combine clusters based on the theory developed in Section \ref{sec:3}. The procedures aim to find the best combination of clusters to maximize the local asymptotic power. 

\subsection{Assumptions and goal} \label{sec:4.2}
Assume the clusters described by $\cJ$ can be classified into two groups such that the identification issue can be solved by combining one cluster from each group. This is formalized in the assumption below.

\begin{assu} \label{assu:4.1}
The set of clusters $\cJ$ is partitioned into nonempty sets $\cJ_0$ and $\cJ_1$ such that:
\begin{enumerate}
	\item \label{assu:4.1.1} $\cJ_0 \cup \cJ_1 = \cJ$.
	\item \label{assu:4.1.2} $\cJ_0 \cap \cJ_1 = \emptyset$.
	\item \label{assu:4.1.3} For any $j_0 \in \cJ_0$ and $j_1 \in \cJ_1$, there is no identification issue in the pair of clusters $\{j_0, j_1\}$.
\end{enumerate}
\end{assu}
In the case where treatments are at the cluster level, as in Examples \ref{eg:3.1} and \ref{eg:3.2}, Assumption \ref{assu:4.1} is automatically satisfied with $\cJ_0$ being the set of control clusters and $\cJ_1$ being the set of treated clusters. For exposition purposes, I also assume that the numbers of clusters in $\cJ_0$ and $\cJ_1$ are equal. Hence, the goal of this section is to pair the clusters. This assumption is relaxed in Section \ref{sec:4.5}.
\begin{assu} \label{assu:4.2} \text{ }
\begin{enumerate}
	\item $|\cJ_0| = |\cJ_1| = \overline q \in \bN$. 
	\item \label{assu:4.2.2} Assume that the researcher is interested in forming $\overline q$ pairs of clusters, where each pair contains one cluster from $\cJ_0$ and one cluster from $\cJ_1$. Each cluster can only be used in one of the pairs. Denote $\Omega$ as the set of all such ways of pairing the clusters.
\end{enumerate}
\end{assu}

Under Assumption {\ref{assu:4.2}.\ref{assu:4.2.2}}, the size of $\Omega$ is $\overline q!$. For each method of combining clusters $\omega \in \Omega$, write $\omega \equiv \{ \{1, \omega_1\}, \{2, \omega_2\}, \ldots, \{\overline q, \omega_{\overline q}\}\} $, where each pair $\{j, \omega_j\}$ represents the clusters $j \in \cJ_0$ and $\omega_j \in \cJ_1$ being grouped together. The following example illustrates the notations.

\begin{eg} Consider $q = 6$. Let $\cJ_0 = \{1, 2, 3\}$ be the control clusters and $\cJ_1 = \{4, 5, 6\}$ be the treated clusters. Then, there are $3! = 6$ different ways to pair the clusters. Here, $\Omega$ collects all different ways of pairing the clusters:
\begin{align*}
	\Omega 
	& = \{ 
		\{\{1, 4\}, \{2, 5\}, \{3, 6\}\}, \ 
			\{\{1, 5\}, \{2, 4\}, \{3, 6\}\}, \ 
			\{\{1, 6\}, \{2, 4\}, \{3, 5\}\}, \\
	& \phantom{=\{\{} 
		\{\{1, 4\}, \{2, 6\}, \{3, 5\}\}, \ 
			\{\{1, 5\}, \{2, 6\}, \{3, 4\}\}, \ 
			\{\{1, 6\}, \{2, 5\}, \{3, 4\}\} \}.
\end{align*}
For instance, if $\omega = \{\{1, 5\}, \{2, 6\}, \{3, 4\}\}$, then $\omega_1 = 5$, $\omega_2 = 6$, and $\omega_3 = 4$ in terms of the notations in the preceding paragraph.
\end{eg}

As mentioned in Section \ref{sec:2.2}, under the scenario where clusters have to be combined, the weak assumptions required on the clusters remain the same by replacing $\cJ$ in Assumption \ref{assu:2.2} with $\omega \in \Omega$. To emphasize the dependence on the pair of clusters $\{j, \omega_j\}$ from $\omega \in \Omega$ and that the pair of clusters are being considered instead of the individual clusters, the notations $n_j$, $\xi_j$, $\widehat\beta_{n,j}$, and $\sigma_{j}$ that are indexed by $j$ in the last two sections are updated to $n_{j, \omega_j}$, $\xi_{j,\omega_j}$, $\widehat\beta_{n,j,\omega_j}$, and $\sigma_{j,\omega_j}$ respectively.

With the above notations, the goal is to solve the following problem in a computationally efficient way:
\begin{align}
	\label{eq:4.1.1}
	\max_{\omega \in \Omega} \quad \pi(\omega, \delta, \alpha),
\end{align}
where $\pi(\omega, \delta, \alpha)$ is the local asymptotic power at significance level $\alpha \in (0, 1)$ with local alternative parameter $\delta \in \bR$ from \eqref{eq:3.1} for a given method of combining clusters $\omega \in \Omega$. \par 
A simple approach to solve the optimization problem \eqref{eq:4.1.1} would be to compute $\pi(\omega, \delta, \alpha)$ for each $\omega \in \Omega$. However, enumerating all possible groupings of clusters and computing $\pi(\omega, \delta, \alpha)$ for each $\omega \in \Omega$ can be computationally intensive when $\overline q$ is large because there are $\overline q!$ different ways in the case of pairing clusters. This motivates the development of computationally efficient procedures for solving \eqref{eq:4.1.1}. \par 

Let $\widehat\pi_n(\omega, \delta, \alpha)$ be the estimator of the objective function in \eqref{eq:4.1.1}. In addition, let $\widehat\omega_n$ be the solution to the sample analog of problem \eqref{eq:4.1.1}, i.e., using $\widehat\pi_n(\omega, \delta, \alpha)$ as the objective in the optimization problem. Note that $\Omega$ remains unchanged when solving the sample analog of \eqref{eq:4.1.1}.
To evaluate $\widehat\pi_n(\omega, \delta, \alpha)$, it requires estimating parameters and variances. This means one has to estimate $\xi_{j,r}$ and $\sigma_{j,r}$ for all $j \in \cJ_1$ and $r \in \cJ_0$. The term $\xi_{j,r}$ is the ratio of cluster size. For $\sigma_{j,r}$, it requires the researcher to specify a working model on the dependence structure and apply a variance estimator. For instance, variance estimators are available for time-dependent data (such as \citet{neweywest1987ecta}) or spatial data (such as \citet{conley1999joe}). Alternatively, one can also estimate variance using quasi-maximum likelihood estimation as in \citet{caoetal2022wp}. Note that even though $\sigma_{j,r}$ has to be estimated, these variance estimators are not used in Algorithm \ref{algo:2.1}. Thus, specifying the working model and the variance estimator does not mean conducting inference using the variance estimator directly. \par 
 
In establishing the validity of the data-driven procedure below, Theorem \ref{prop:4.7} ahead does not require the correct specification of the working model. However, the choice of the working model can affect the power performance as it affects the objective function of \eqref{eq:4.1.1}. \par

To justify the validity of the test that combines the clusters according to the optimal solution to the sample analog of \eqref{eq:4.1.1}, consider the following assumption that requires how the clusters are combined has a negligible impact on the test. This is similar to the condition imposed in \citet{caoetal2022wp}.

\begin{assu} \label{assu:4.4}
There exists $\widetilde{\Omega} \subseteq \Omega$ such that the following holds:
\begin{enumerate}
	\item \label{assu:4.4.1} $\displaystyle \lim_{n \to\infty} \bP[\widehat\omega_n \notin \widetilde\Omega] = 0$.
	\item \label{assu:4.4.2} $\displaystyle \lim_{n\to\infty} \sup_{\omega \in \widetilde{\Omega}} |\bP[\phi_n(\widehat\omega_n) = 1 | \widehat\omega_n = \omega] - \bP[\phi_n(\omega) = 1]| = 0$.
\end{enumerate}
\end{assu}

The above assumption can be satisfied in different settings. For example, this holds through sample-splitting where part of the data is used to obtain $\widehat\omega_n$, and another part is used to conduct inference. Note that Assumption \ref{assu:4.4} is also satisfied under some weak assumptions as follows.

\begin{assu} \label{assu:4.5}
Let $\delta \in \bR$ be a local alternative parameter as in \eqref{eq:3.1} and $\alpha \in (0, 1)$. Assume that the following statements hold:
\begin{enumerate}
	\item \label{assu:4.5.1} $\widehat\pi_n(\omega, \delta, \alpha) \Pt \pi(\omega, \delta, \alpha)$ for each $\omega \in \Omega$.
	\item \label{assu:4.5.2} $\widetilde\Omega = \{\omega^\star\}$, where $\omega^\star \in \Omega$.
	\item \label{assu:4.5.3} $\pi(\omega^\star, \delta, \alpha) > \pi(\omega, \delta, \alpha)$ for each $\omega \in \Omega \backslash \{\omega^\star\}$.
\end{enumerate}
\end{assu}
Assumption \ref{assu:4.5}.\ref{assu:4.5.1} is a mild condition requiring the estimator to converge for each $\omega \in \Omega$. As before, this assumption does not require correctly specifying the local asymptotic power function. Assumptions \ref{assu:4.5}.\ref{assu:4.5.2} and \ref{assu:4.5}.\ref{assu:4.5.3} require that $\omega^\star  \in \Omega$ is the unique solution to \eqref{eq:4.1.1}. The following proposition shows that Assumption \ref{assu:4.4} is satisfied under these conditions.

\begin{prop} \label{prop:4.6b}
Consider the problem of testing \eqref{eq:2-hypo}. Let Assumption \ref{assu:2.2} hold with $\cJ$ replaced by any $\omega \in \Omega$, Assumptions \ref{assu:4.1} and \ref{assu:4.2} hold. Then, Assumption \ref{assu:4.5} implies that Assumption \ref{assu:4.4} holds.
\end{prop}

The following theorem establishes the validity of the test that chooses $\widehat\omega_n \in \Omega$ by solving problem \eqref{eq:4.1.1}.

\begin{thm} \label{prop:4.7}
Consider the problem of testing \eqref{eq:2-hypo}. Let Assumption \ref{assu:2.2} hold with $\cJ$ replaced by any $\omega \in \Omega$, Assumptions  \ref{assu:3.3.4}, \ref{assu:4.1}, \ref{assu:4.2}, and \ref{assu:4.4} hold, and $\alpha \in (0, 1)$. Let $\phi_n(\omega)$ be the test in \eqref{eq:test-1} that uses $\omega \in \Omega$ to combine the cluster. Then, 
\[
	\lim_{n\to\infty} \bP[\phi_n(\widehat\omega_n) = 1] \leq \alpha,
\]
under the null hypothesis. 
\end{thm}

The above theorem shows that the data-driven procedure controls size under some high-level conditions when the sample analog of problem \eqref{eq:4.1.1} is used. The validity of the procedure does not require the local asymptotic power function to be correctly specified, although having an incorrect specification can lead to power loss.

\subsection{Procedure for $K=1$} \label{sec:k1}

This section develops a procedure to combine clusters for $K = 1$, i.e., when the significance level is chosen such that $\alpha \in [\frac{1}{2^{\overline q-1}}, \frac{1}{2^{\overline q - 2}} )$. The procedure utilizes the local asymptotic power derived in Corollary \ref{cor:3.3.5}. Using the notations from Section \ref{sec:4.2}, the local asymptotic power for $K = 1$ with a given $\omega \equiv \{ \{1, \omega_1\}, \{2, \omega_2\}, \ldots, \{\overline q, \omega_{\overline q}\}\} \in \Omega$ and $\delta \in \bR$ can be written as
\begin{equation}
	\label{eq:4.2.1}
	\pi(\omega, \delta, \alpha)
	\equiv \prod^{\overline q}_{j=1} 
		\Phi \left( -\frac{\xi_{j, \omega_j} \delta}{\sigma_{j, \omega_j}}	\right)
	+ \prod^{\overline q}_{j=1} 
		\left[
		1-\Phi \left( -\frac{\xi_{j, \omega_j} \delta}{\sigma_{j, \omega_j}}	\right)
		\right],
\end{equation}
for $\alpha \in [\frac{1}{2^{\overline q-1}}, \frac{1}{2^{\overline q - 2}} )$.
 
To develop a computationally feasible approach for solving optimization problem \eqref{eq:4.1.1} with objective \eqref{eq:4.2.1}, recall from Assumption \ref{assu:4.2} that each cluster in $j \in \cJ_0$ is matched with a cluster in $r \in \cJ_1$. Moreover, each cluster from $\cJ_0$ and $\cJ_1$ is matched only once. As a result, there are only $\overline q$ possible terms $\sigma_{j, r}$ and $\xi_{j,r}$ for each $j \in \cJ_0$. Hence, there are $\overline q^2$ possible terms to compute in total. This is less than $\overline q!$ when $\overline q \geq 4$. See the example below.

\begin{eg} Consider $q = 8$. Let $\cJ_0 = \{1, 2, 3, 4\}$ be the control clusters and $\cJ_1 = \{5, 6, 7, 8\}$ be the treated clusters. For each $j \in \cJ_0$, there are four distinct parameters $\sigma_{j,r}$ to precompute over $r \in \cJ_1$. Hence, there are $4^2=16$ terms to precompute. Similarly, $\xi_{j,r}$ also has 16 possible combinations.
\end{eg}

Let $\Psi_{j, r} \equiv \Phi \left( -\frac{\xi_{j, r} \delta}{\sigma_{j, r}} \right)$ for each $j \in \cJ_0$ and $r \in \cJ_1$. There are $\overline q^2$ possible combinations of this term by the above reasoning. Following the above discussion, these terms can also be precomputed in advance. Let $z_{j, r} \in \{0, 1\}$ be an indicator variable that equals 1 if clusters $j \in \cJ_0$ and $r \in \cJ_1$ are paired together, and equals 0 otherwise. Then, the optimization problem \eqref{eq:4.1.1} with objective \eqref{eq:4.2.1} can be formulated as the following integer program:
\begin{align}
	\label{eq:4.2.3}
	\begin{split}
	\max_{z_{j, r}}
	\quad &
		\prod_{j = 1}^{\overline q} \sum_{r = 1}^{\overline q} \Psi_{j, r} z_{j, r} +
		\prod_{j = 1}^{\overline q} \sum_{r = 1}^{\overline q} \left(1 - \Psi_{j, r} \right) z_{j, r} \\
		  \text{s.t.} \quad 
		   & \sum_{j = 1}^{\overline q} z_{j, r} = 1 \hspace{57.3pt} \text{ for each $j =1,\ldots, \overline q$}	\\
		  & \sum_{r = 1}^{\overline q} z_{j, r} = 1 \hspace{57.3pt} \text{ for each $r=1,\ldots, \overline q$}	\\
		   & z_{j, r} \in \{0, 1\}	\qquad \hspace{30pt} \text{ for each $j=1,\ldots, \overline q$ and $r =1,\ldots, \overline q$}.
	\end{split}
\end{align}

While the optimization problem \eqref{eq:4.2.3} has transformed the optimization problem \eqref{eq:4.1.1} with objective \eqref{eq:4.2.1} into an integer program with linear constraints, the objective function is nonlinear in $z_{j,r}$. To obtain a computationally feasible program, recall from Property \ref{ppt:3.7}.\ref{ppt:3.7.3} that if $\delta < 0$, then
\begin{equation}
	\label{eq:4.2.4}
	\prod_{j = 1}^{\overline q} \sum_{r = 1}^{\overline q} \Psi_{j, r} z_{j, r}
	 > \frac{1}{2^{\overline q}}
	 > \prod_{j = 1}^{\overline q} \sum_{r = 1}^{\overline q} \left(1 - \Psi_{j, r} \right) z_{j, r} ,
\end{equation}
where the relationship is reversed for $\delta > 0$. Since $z_{j, r}$ is a binary variable, taking logarithm on inequality \eqref{eq:4.2.4} gives
\begin{equation}
	\label{eq:4.2.5}
	\sum^{\overline q}_{j=1} \sum^{\overline q}_{r = 1}
	z_{j, r} \log \Psi_{j, r}
	> -\overline q \log 2 
	> 
	\sum^{\overline q}_{j=1} \sum^{\overline q}_{r = 1}
	z_{j, r} \log (1 - \Psi_{j, r}).
\end{equation}
Note that $\log \Psi_{j,r}$ and $\log (1 - \Psi_{j,r})$ have been pre-computed in advance. Hence, the terms on the left-hand side and the right-hand side of \eqref{eq:4.2.5} are linear in the unknowns $z_{j,r}$. \par 
Next, inequality \eqref{eq:4.2.4} implies that the RHS is in the small interval $[0, \frac{1}{2^{\overline q}}]$. The interval $[0, \frac{1}{2^{\overline q}}]$ can be partitioned into $A$ subintervals $[\epsilon_0, \epsilon_1], [\epsilon_1, \epsilon_2], \ldots, [\epsilon_{A-1}, \epsilon_A]$ where $0 < \epsilon_0 < \epsilon_1 < \cdots < \epsilon_A = \frac{1}{2^{\overline{q}}}$. Since the log of RHS of \eqref{eq:4.2.4} is considered in \eqref{eq:4.2.5}, the number $\epsilon_0$ can be chosen to be $\min_{j, r \in \{1, \ldots, \overline q\}}  (1 - \Psi_{j,r})^{\overline q}$ or a very small number above 0 in order to have a finite value for $\log \epsilon_0$. \par 
If $\delta < 0$, the optimization problem \eqref{eq:4.2.3} can be approximated as solving multiple integer linear programs where each program maximizes the left-hand side of \eqref{eq:4.2.5} and restricts the right-hand side of \eqref{eq:4.2.5} in a small interval $[\log \epsilon_{a-1}, \log \epsilon_a]$ for each $a = 1, \ldots, A$:
\begin{align}
	\label{eq:4.2.6}
	\begin{split}
		\max_{z_{j, r}} \quad 
		&  \sum_{j = 1}^{\overline q} \sum_{r = 1}^{\overline q} z_{j, r} \log \Psi_{j, r} \\ 
		  \text{s.t.} \quad 
		   & \sum_{j = 1}^{\overline q} z_{j, r} = 1 \qquad\qquad \text{ for each $j = 1,\ldots, \overline q$}	\\
		  & \sum_{r = 1}^{\overline q} z_{j, r} = 1 \qquad\qquad \text{ for each $r = 1,\ldots, \overline q$}	\\
		   &  \sum_{j = 1}^{\overline q} \sum_{r = 1}^{\overline q} z_{j, r} \log\left(1 - \Psi_{j, r} \right) \in [\log\epsilon_{a-1}, \log\epsilon_{a}] \\ 
		   & z_{j, r} \in \{0, 1\}	\qquad\qquad \text{for each $j = 1,\ldots, \overline q$ and $r = 1,\ldots, \overline q$.}
	\end{split}
\end{align}

Since \eqref{eq:4.2.6} is merely a linear program with binary variables, solving $A$ of them is fast even for large values of $A$. After solving the optimization problem \eqref{eq:4.2.6} for each $a = 1, \ldots, A$, the solution can be chosen to be the one that gives the largest local asymptotic power. In practice, \eqref{eq:4.2.6} is replaced with its sample analog, i.e., by replacing $\Psi_{j,r}$ with $\widehat{\Psi}_{j,r}$ for each $j \in \cJ_1$ and $r \in \cJ_0$. This is done by estimating $\xi_{j,r}$ and $\sigma_{j,r}$ as discussed in Section \ref{sec:4.2}.\par

If $\delta > 0$, the procedure is the same except the 0-1 linear integer program in \eqref{eq:4.2.6} becomes the following program:
\begin{align}
	\label{eq:4.2.6b}
	\begin{split}
		\max_{z_{j, r}} \quad 
		&  \sum_{j = 1}^{\overline q} \sum_{r = 1}^{\overline q} z_{j, r} \log \left(1 - \Psi_{j, r} \right) \\ 
		  \text{s.t.} \quad 
		   & \sum_{j = 1}^{\overline q} z_{j, r} = 1 \qquad\qquad \text{ for each $j = 1,\ldots, \overline q$}	\\
		  & \sum_{r = 1}^{\overline q} z_{j, r} = 1 \qquad\qquad \text{ for each $r = 1,\ldots, \overline q$}	\\
		   &  \sum_{j = 1}^{\overline q} \sum_{r = 1}^{\overline q} z_{j, r} \log \Psi_{j, r}  \in [\log\epsilon_{a-1}, \log\epsilon_{a}] \\ 
		   & z_{j, r} \in \{0, 1\}	\qquad\qquad \text{for each $j = 1,\ldots, \overline q$ and $r = 1,\ldots, \overline q$.}
	\end{split}
\end{align}

 The procedure described in this section is summarized as follows.

\begin{algo} \label{algo:4.2.1} \text{ }
\begin{description}
	\item[Step 1:] Choose the local alternative parameter $\delta \in \bR$. Estimate and precompute $\Psi_{j, r}$ for each $j, r = 1, \ldots, \overline q$. There are $\overline q^2$ terms to precompute. 
	\item[Step 2:] Form $A$ subintervals of $[0, \frac{1}{2^{\overline q}}]$ as $[\epsilon_0, \epsilon_1], [\epsilon_1, \epsilon_2], \ldots, [\epsilon_{A-1}, \epsilon_A]$ where $0 < \epsilon_0 < \epsilon_1 < \cdots < \epsilon_A = \frac{1}{2^{\overline{q}}}$ and $\epsilon_0$ is $\min_{j, r \in \{1, \ldots, \overline q\}} (1 - \Psi_{j,r})^{\overline q}$ or a smaller positive number.
	\item[Step 3:] If $\delta < 0$, solve the sample analog of \eqref{eq:4.2.6}. If $\delta > 0$, solve the sample analog of \eqref{eq:4.2.6b}. Solve the program for each subinterval $[\log \epsilon_{a-1}, \log\epsilon_{a}]$ where $a=1,\ldots, A$. Let $\omega^{(a)}$ be the solution and $\pi^{(a)}$ be the local asymptotic power for the $a$th problem. Set $\pi^{(a)} = -\infty$ if the $a$th program is infeasible.
	\item[Step 4:] Return $\omega^{(a^\star)}$ as the solution such that $a^\star = \argmax_{a = 1, \ldots, A} \pi^{(a)}$.
	\item[Step 5:] Conduct inference using CRS as in Algorithm \ref{algo:2.1} with  $\omega^{(a^\star)}$ as the clusters.
\end{description}
\end{algo}

Algorithm \ref{algo:4.2.1} requires the researcher to choose a value of $A$ in Step 2. The following proposition establishes that when $A$ is large enough, then the approximation in Algorithm \ref{algo:4.2.1} yields the same solution as in the original optimization problem \eqref{eq:4.2.3}. Hence, by partitioning $[0, \frac{1}{2^{\overline q}}]$ into fine enough intervals, the solution from the approximation step coincides with the solution from \eqref{eq:4.2.3}. 

\begin{prop} \label{prop:4.9}
Let $\pi^\star$ be the optimal value to optimization problem \eqref{eq:4.2.3}. There exists $\epsilon_0 > 0$ and a positive integer $A_0$ such that if $[\epsilon_0, \frac{1}{2^{\overline q}}]$ is partitioned to $A_0$ equally-spaced intervals, then the optimal value that corresponds to $a^\star$ in Step 4 of Algorithm \ref{algo:4.2.1} is $\pi^\star$.
\end{prop}

The above proposition assumes the researcher partitions the interval $[0, \frac{1}{2^{\overline q}}]$ into equally-spaced intervals. This assumption is for convenience purposes only as it seemed to be a natural way to run Algorithm \ref{algo:4.2.1}. \par 

This section ends with three remarks on the above data-driven procedure for combining clusters.

\begin{re}[Log-linearization of the objective function.] \label{rmk:4.9}
Instead of solving $A$ 0-1 integer linear programs as in Algorithm \ref{algo:4.2.1}, one can approximate the objective function of \eqref{eq:4.2.3} by taking the sum of the left-hand side and the right-hand side of \eqref{eq:4.2.5} as the objective function. This gives the following 0-1 integer linear program:
\begin{align}
	\label{eq:4.2.7}
	\begin{split}
		\max_{z_{j, r}} \quad 
		&  \sum^{\overline q}_{j=1} \sum^{\overline q}_{r = 1}
	z_{j, r} \log \Psi_{j, r}
	+
	\sum^{\overline q}_{j=1} \sum^{\overline q}_{r = 1}
	z_{j, r} \log (1 - \Psi_{j, r}) \\ 
		  \text{s.t.} \quad 
		   & \sum_{j = 1}^{\overline q} z_{j, r} = 1 \qquad \hspace{27.3pt} \text{ for each $j = 1,\ldots, \overline q$}	\\
		  & \sum_{r = 1}^{\overline q} z_{j, r} = 1 \qquad \hspace{27.3pt}  \text{ for each $r = 1,\ldots, \overline q$}	\\
		   & z_{j, r} \in \{0, 1\}	\qquad \qquad \text{ for each $j = 1,\ldots, \overline q$ and $r = 1,\ldots, \overline q$.}
	\end{split}
\end{align}
Optimization problem \eqref{eq:4.2.7} has the advantage of only requiring to solve one optimization problem when compared to \eqref{eq:4.2.6}. However, the objective function of \eqref{eq:4.2.7} may not be a good approximation of \eqref{eq:4.2.3} because it weights the two summations equally. In fact, simulations show this can perform worse than randomly picking a grouping of clusters in each draw of data.
\end{re}

\begin{re}[Homogeneity of clusters] \label{re:4.homo} When the control clusters are homogeneous such that $\Psi_{j,r} = \Psi_j$ across all $r \in \cJ_1$ for each $j \in \cJ_0$, how the clusters are combined is not going to affect the local asymptotic power. The same is also true if the treated clusters are homogeneous such that $\Psi_{j, r} = \Psi_r$ across all $j \in \cJ_0$ for each $r \in \cJ_1$.
\end{re}

\begin{re}[Step 3 of Algorithm \ref{algo:4.2.1}]
If $\delta < 0$, it is possible that for some $a = 1, \ldots, A$, optimization problem \eqref{eq:4.2.6} is infeasible. This happens when there are no combinations of $z_{j,r}$ such that the following constraints hold:
\begin{equation}
	\label{eq:re-4.12}
	\sum_{j = 1}^{\overline q} \sum_{r = 1}^{\overline q} z_{j, r} \log\left(1 - \Psi_{j, r} \right) \in [\log\epsilon_{a-1}, \log\epsilon_{a}] ,
\end{equation}
for a particular value of $a$. However, there must exist at least one $a' = 1, \ldots, A$ such that \eqref{eq:re-4.12} holds as long as $\epsilon_0$ is small enough. This is because the summation in \eqref{eq:re-4.12} is bounded above by $-\overline q \log 2$ by construction in \eqref{eq:4.2.5}. If $\delta > 0$, the same discussion applies  to the constraint 
\[
	\sum_{j = 1}^{\overline q} \sum_{r = 1}^{\overline q} z_{j, r} \log \Psi_{j, r}  \in [\log\epsilon_{a-1}, \log\epsilon_{a}]
\]
in optimization problem \eqref{eq:4.2.6b}.
\end{re}

\subsection{Procedure for $K>1$}

The previous section focuses on developing an algorithm that uses the local asymptotic power of $K = 1$ as the objective function. In practice, researchers may be interested in conducting a test that involves $K > 1$. The reason for having a simple 0-1 integer linear program in the previous section is because the local asymptotic power for $K = 1$ has a simple expression as stated in Corollary \ref{cor:3.3.5}. The other cases require evaluating the expression in Theorem \ref{prop:3.3.4} that involves more complicated expressions but can be evaluated using numerical methods. \par 

Similar to the discussion in Section \ref{sec:k1} on the estimation of variances, evaluating the rejection probability in Theorem \ref{prop:3.3.4} requires the researcher to specify a working model on the dependence structure and apply a variance estimator. Recall that Theorem \ref{prop:4.7} does not impose any requirement on $K$. Hence, the choice of the working model and the variance estimator does not affect the validity of the procedure, but can affect the power performance.

Algorithm \ref{algo:4.2.1} is found to be effective in computing the optimal grouping of clusters for $K > 1$ in simulations. Hence, it can be used to obtain an initial solution. When $\overline q$ is large and $K>1$, evaluating the local asymptotic power for each $\omega \in \Omega$ can be costly. The following algorithm provides a heuristic to proceed from the initial solution. The main idea is to start from the initial point provided by Algorithm \ref{algo:4.2.1}. Then, one proceeds to see if it is possible to improve upon the existing solution using local search.  This algorithm below is motivated by the 2-opt algorithm for the traveling salesman problem (see, for instance, \citet{flood1956or}, \citet{croes1958or}, and \citet{linkernighan1973or}).

\begin{algo} \label{algo:4.2.2} \text{ } 
\begin{description}
	\item[Step 1:] Obtain an initial solution $\omega^{(1)} \equiv \{\{1, \omega_1^{(1)}\}, \{2, \omega_2^{(1)}\}, \ldots, \{\overline q, \omega_{\overline q}^{(1)}\}\}$ by running Algorithm \ref{algo:4.2.1}. Let the corresponding local asymptotic power of the initial solution be $\pi^{(1)}$.
	\item[Step 2:] Swap the assignment of every two pairs of clusters in $\omega^{(1)}$ and compute the local asymptotic power after the swap. This means picking two distinct pairs of clusters in $\omega^{(1)}$, e.g., $\{i, \omega_i^{(1)}\}$ and $\{j, \omega_j^{(1)}\}$, swap their assignment as $\{i, \omega_j^{(1)}\}$ and $\{j, \omega_i^{(1)}\}$, keep the other assignments the same and recompute the local asymptotic power.
	\item[Step 3:] If a better solution is found, i.e., there exists $\omega^{(2)}$ with local asymptotic power $\pi^{(2)}$ such that $\pi^{(2)} > \pi^{(1)}$, repeat Step 2 with $\omega^{(1)}$ replaced by $\omega^{(2)}$. Stop if there is no solution that results in a higher local asymptotic power.
	\item[Step 4:] Let $\omega^\star$ be the solution to the problem. Conduct inference using CRS as in Algorithm \ref{algo:2.1} using  $\omega^\star$ as the clusters.
\end{description}
\end{algo}

Algorithm \ref{algo:4.2.2} presents a procedure by conducting pairwise swaps. This can be easily extended for swapping more groupings.

\subsection{Extension} \label{sec:4.5}

This section briefly discusses the case with an unequal number of clusters in $\cJ_0$ and $\cJ_1$ for Algorithm \ref{algo:4.2.1}. Assume without loss of generality that $|\cJ_0| < |\cJ_1|$. \par 

Following the notations in Section \ref{sec:4.2}, let $\Omega$ be the collection of all methods of combining clusters. In addition, let $\cJ_0 = \{1, 2, \ldots, \oq\}$. Hence, the goal is to form $\overline{q}$ groups of clusters, where each group has one cluster from $\cJ_0$, and at least one cluster from $\cJ_1$. Moreover, all clusters from $\cJ_1$ have to be used. As before, each cluster can only be used once in the grouping. Therefore, each $\omega \in \Omega$ can be written as $\{\{1, \omega_1\}, \{2, \omega_2\}, \ldots, \{\oq, \omega_{\oq}\}\}$ as before, with the exception that $\omega_j$ can represent more than one cluster from $\cJ_1$ here for each $j = 1, \ldots, \overline{q}$. Let $\cM$ be the set of all possible ways of partitioning clusters in $\cJ_1$ into $\overline{q}$ groups, so that $\omega_j \in \cM$ for each $j = 1, \ldots, \oq$, and that each $m \in \cM$ is nonempty and satisfies $m \subset \cJ_1$.  For example, if $\cJ_0 = \{1, 2\}$ and $\cJ_1 = \{3, 4, 5\}$ and one wishes to form two groups of clusters, the set $\cM$ can be written as $\cM = \{\{3\}, \{4\}, \{5\}, \{3, 4\}, \{3, 5\}, \{4, 5\}\}$.

Similar to Section \ref{sec:k1}, let $z_{j,m} \in \{0, 1\}$ be a binary variable that equals 1 when cluster $j \in \cJ_0$ is combined with $m \in \cM$. Define $\Psi_{j,m} \equiv \Psi \left( -\frac{\xi_{j,m}\delta}{\sigma_{j,m}}\right)$, where $\xi_{j,m}$ is the ratio of cluster size and $\sigma^2_{j,m}$ is the asymptotic variance in the group of cluster $\{j, m\}$, where $j \in \cJ_0$ and $m \in \cM$.

For each $a = 1, \ldots, A$ and the corresponding interval $[\log\epsilon_{a-1}, \log\epsilon_a]$, the optimization problem analogous to problem \eqref{eq:4.2.6} for $\delta < 0$ in the scenario of $|\cJ_0| < |\cJ_1|$ is as follows:
\begin{align*}
		\max_{z_{j, m}} \quad 
		&  \sum_{j = 1}^{\overline q} \sum_{m \in \cM} z_{j, m} \log \Psi_{j, m} \\ 
		  \text{s.t.} \quad 
		   & \sum_{m \in \cM} z_{j, m} = 1 && \text{ for each $j = 1,\ldots, \oq$}	\\
		  & \sum_{j = 1}^{\overline q} z_{j, m} = 1 && \text{ for each $m \in \cM$}	\\
		   &  \sum_{j = 1}^{\overline q} \sum_{m \in \cM} z_{j, m} \log\left(1 - \Psi_{j, m} \right) \in [\log\epsilon_{a-1}, \log\epsilon_{a}] \\ 
		   & z_{j, m} \in \{0, 1\}	&& \text{ for each $j = 1,\ldots, \overline q$ and $m \in \cM$.}
\end{align*}

Similarly, if $\delta > 0$, the optimization problem becomes
\begin{align*}
		\max_{z_{j, m}} \quad 
		&  \sum_{j = 1}^{\overline q} \sum_{m \in \cM} z_{j, m} \log \left(1 - \Psi_{j, m}\right) \\ 
		  \text{s.t.} \quad 
		   & \sum_{m \in \cM} z_{j, m} = 1 && \text{ for each $j = 1,\ldots, \oq$}	\\
		  & \sum_{j = 1}^{\overline q} z_{j, m} = 1 && \text{ for each $m \in \cM$}	\\
		   &  \sum_{j = 1}^{\overline q} \sum_{m \in \cM} z_{j, m} \log  \Psi_{j, m}  \in [\log\epsilon_{a-1}, \log\epsilon_{a}] \\ 
		   & z_{j, m} \in \{0, 1\}	&& \text{ for each $j = 1,\ldots, \overline q$ and $m \in \cM$.}
\end{align*}
The above optimization problems are binary integer programs as Section \ref{sec:k1}, and are fast to solve.

\section{Monte Carlo simulations}	\label{sec:5}

\subsection{Data generating process}
This section explores the finite-sample performance of the data-driven procedures to combine clusters in Section \ref{sec:4}. Here, I consider the data generating process (DGP) based on the simulation design in \citet{hagemann2022wp}, which is a version of the simulation design in \citet{conleytaber2011restat}. \par 

Let there be $q = 12$ clusters, which consists of 6 treated clusters represented by $\cJ_1 = \{1, 2, \ldots, 6\}$ and 6 control clusters represented by $\cJ_0 = \{7, 8, \ldots, 12\}$. In the following, $j \in \cJ \equiv \cJ_1 \cup \cJ_0$ indices all the clusters and $t \in \cT \equiv \{1, 2, \ldots, T\}$ indices time.  Let $D_{t,j} \equiv \ind[j \in \cJ_1 \text{ and } t > t_0]$ for $j \in \cJ$ and $t \in \cT$ and $I_t \equiv \ind[t > t_0]$ for $t \in \cT$. Each simulated data is generated from the following model:
\begin{align*}
	Y_{t,j} & = \theta_0 I_t + \beta D_{t,j} + \gamma_1 X_{1,t,j} + \gamma_2 X_{2,t,j} + \gamma_3 X_{3,t,j} + \xi_j + U_{t,j}, \\
	U_{t,j} & = \rho U_{t-1, j} + V_{t,j}, \\ 
	X_{1,t,j} & = \gamma_4 I_t D_{j,t} + W_{t,j},
\end{align*}
with $\theta_0 = \gamma_1 = \gamma_2 = \gamma_3 =  \xi_j = 1$, $\rho = 0.5$, $\gamma_4 = 0.8$, $t_0 = 10$, and $T = 20$ as in \citet{hagemann2022wp}. I consider three different DGPs below. DGP 1 generates the random variables in the same way as \citet{hagemann2022wp}. DGPs 2 and 3 have different variance structures and assume there is the same number of heterogenous treated and control clusters. In all the specifications, the parameter $h \in \{1, 2, 3, 4\}$ governs the amount of heterogeneity in data, where a larger value of $h$ represents a higher degree of heterogeneity. The goal is to test the two-sided hypothesis $H_0 : \beta = 0$ vs $H_1 : \beta \neq 0$. I consider two significance levels $\alpha = 0.05$ and $\alpha = 0.1$ in order to study the finite-sample properties of Algorithms \ref{algo:4.2.1} and \ref{algo:4.2.2} separately. The specification for each $j \in \cJ$ are as follows.

\begin{description}
	\item[DGP 1.] $X_{2,t,j}, X_{3,t,j}, V_{t,j}, W_{t,j} \iid N(0, \sigma_j^2)$, with 
	\[
		\sigma_j 
		= 
		\begin{cases}
			20 & , \ \ \  j \geq 13 - h	\\
			1 & ,  \ \ \ \text{otherwise}
		\end{cases}.
	\]
	\item[DGP 2.] Same as DGP 1, but	
	\[
		\sigma_j 
		= 
		\begin{cases}
			5 + 3(j \text{ mod } 6) & , \ \ \  (j \text{ mod } 6) \leq h - 1	\\
			1 & ,  \ \ \ \text{otherwise}
		\end{cases}.
	\]
	\item[DGP 3.] Same as DGP 1, but	
	\[
		\sigma_j 
		= 
		\begin{cases}
			2.5^{1 + (j \text{ mod } 6)} & , \ \ \  (j \text{ mod } 6) \leq h - 1	\\
			1 & ,  \ \ \ \text{otherwise}
		\end{cases}.
	\]
\end{description}

For each simulation design, I estimate the rejection rate obtained from the procedure in Section \ref{sec:4} and use CRS to conduct inference for each of the $6!=720$ different methods of combining clusters. Apart from CRS, I also apply other recent methods for conducting inference with a fixed number of clusters that were mentioned in Section \ref{sec:intro}. The methods and their abbreviations are summarized in Table \ref{tab:5.abbrev}. Some details of considering the two versions of CRS in Table \ref{tab:5.abbrev} are as follows. First, \textbf{CRS-Data} is the data-driven procedure that uses optimization problem \eqref{eq:4.2.6} for $\delta < 0$ and problem \eqref{eq:4.2.6b} for $\delta > 0$ when $\alpha = 0.05$. For these two optimization problems, I partition the interval $[0, \frac{1}{2^5}]$ into $A = 200$ subintervals. When $\alpha = 0.1$, the heuristic in Algorithm \ref{algo:4.2.2} is used. One of the two programs  \eqref{eq:4.2.6} and \eqref{eq:4.2.6b} is used to obtain an initial solution, depending on the sign of $\delta$. Next, \textbf{CRS-Random} chooses one of the 720 ways of combining clusters randomly in each draw. This aims to represent the scenario where the researcher chooses a random grouping of clusters that can solve the identification problem for each data set. For each specification, I use 20,000 Monte Carlo simulations. The data-driven procedure \textbf{CRS-Data} uses the local alternative parameter $\delta = 2\sqrt{q T}$ under $\beta \geq 0$ and $\delta = -2\sqrt{q T}$ under $\beta < 0$. I consider $\beta \in \{-3, -2, \ldots, 3\}$ in the simulation exercises.

\begin{table}[!ht]
	\centering
	\caption{Abbreviations for the methods.}
	\label{tab:5.abbrev}
	\renewcommand{\arraystretch}{1.2}
	\begin{tabular}{ll}
		\toprule
		\multicolumn{2}{l}{\textbf{\emph{Various versions of CRS}}} \\
		\bf CRS-Data & The data-driven procedure as in Algorithm \ref{algo:4.2.1} for $\alpha = 0.05$ and \\
		& Algorithm \ref{algo:4.2.2} for $\alpha = 0.1$.	\\
		\bf CRS-Random & Randomly choose one way of combining clusters in each draw of data.	\\
		\midrule
		\multicolumn{2}{l}{\textbf{\emph{Other methods}}} \\
		\bf BCH &  The test in \citet{besteretal2011joe} that uses the cluster covariance matrix \\
		&  estimator.	\\
		\bf H  & The adjusted permutation test by \citet{hagemann2022wp}.	\\
		\bf IM  & The $t$-test in \citet{ibraimovmuller2016restat}.	\\
		\bf WCB  & The Wild cluster bootstrap \citep{cameronetal2008restat}.	\\
		\bottomrule
	\end{tabular}
\end{table}

\subsection{Simulation results}

Figure \ref{fig:sim-005-new} presents the rejection rate curves for various DGPs, heterogeneity parameters $h$, and methods for conducting inference at significance level $\alpha=0.05$. Each column corresponds to one of the three DGPs, with different amounts of heterogeneity. In each figure, the grey region is formed by overlapping 720 different rejection rate curves based on each of the different ways of combining clusters, and conducting inference using CRS. The other colored lines correspond to various methods as indicated in the legend with the abbreviations specified in Table \ref{tab:5.abbrev}. \par 

For DGP 1 (column 1 of Figure \ref{fig:sim-005-new}), the treated clusters are homogeneous for the different values of $h$ because the heterogeneity in variance only affects the control clusters. Hence, how the clusters are combined should not affect the population local asymptotic power as discussed in Remark \ref{re:4.homo}. Indeed, both \textbf{CRS-Data} and \textbf{CRS-Random} give similar performance in finite-samples. On the other hand, BCH and WCB over-rejects under the null when $h = 4$.\par 

For the other DGPs (columns 2 and 3 of Figure \ref{fig:sim-005-new}), there is heterogeneity in both the treated and control clusters. Hence, there is more variation in the CRS rejection rates depending on how the clusters are combined. CRS controls size in all designs. In addition,  \textbf{CRS-Data} (as indicated by the solid black lines) always belong to the top regions formed by the grey curves. This shows that Algorithm \ref{algo:4.2.1} can lead to close to oracle performance in the simulations. Besides, \textbf{CRS-Data} also performs better than  \textbf{CRS-Random} (as indicated by the solid blue lines). This shows that the data-driven procedures using \eqref{eq:4.2.6} or \eqref{eq:4.2.6b} perform better than choosing a random grouping of clusters. On the other hand, WCB can over-reject in designs with more heterogeneity. 

\begin{figure}[!ht]
	\centering
	\caption{Rejection rate curves for the simulations at significance level $\alpha=0.05$.}
	\label{fig:sim-005-new}
	\includegraphics[scale=1]{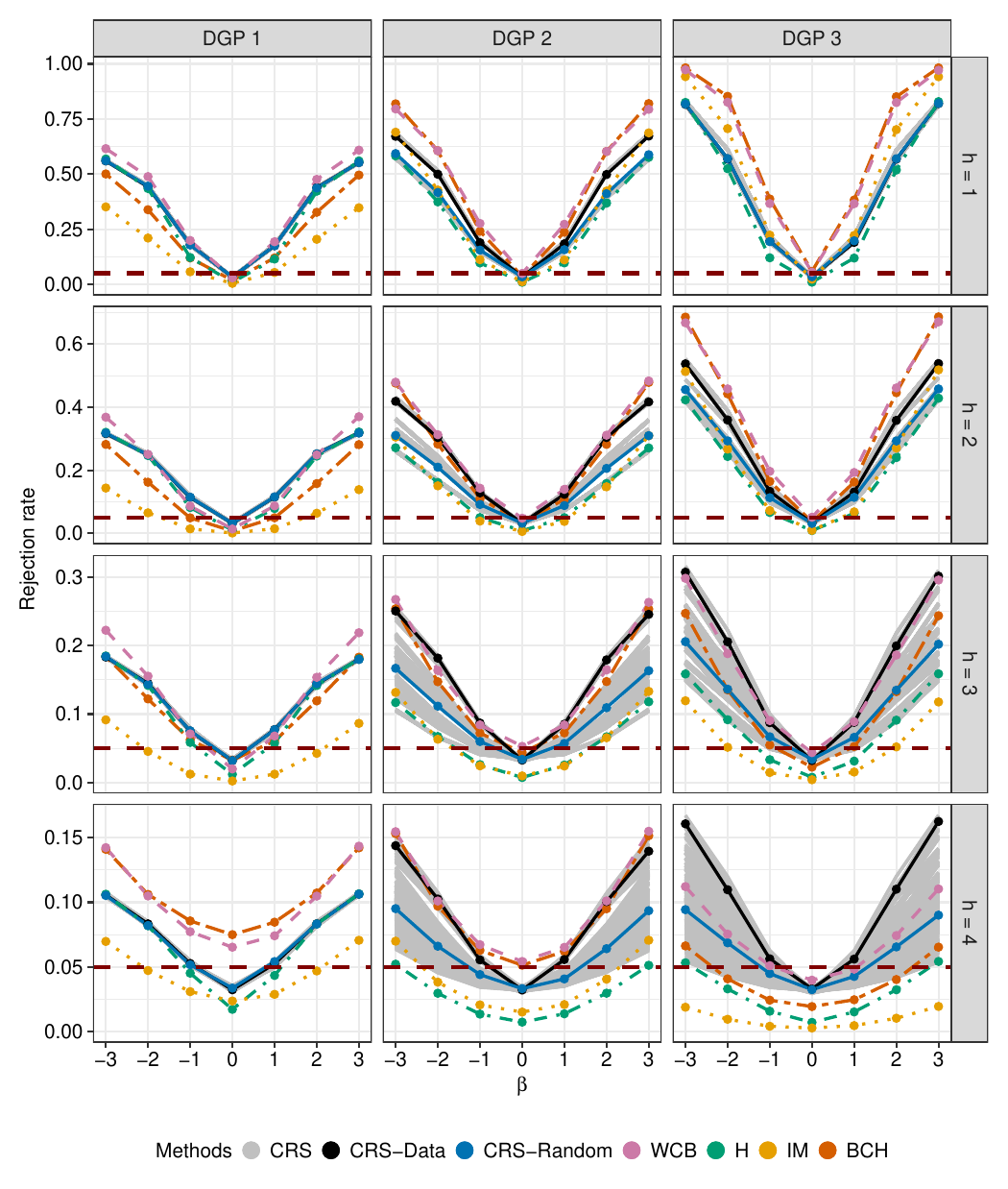}
\end{figure}

Next, I examine the performance of Algorithm \ref{algo:4.2.2} by considering $\alpha = 0.1$ for the same set of simulations. The results are summarized in Figure \ref{fig:sim-010-new}. The observations for $\alpha = 0.1$ are similar to those when $\alpha = 0.05$ as in Figure \ref{fig:sim-005-new}. CRS controls size in all designs. \textbf{CRS-Data} continues to be able to lead to a close-to-oracle grouping of clusters as the solid black lines belong to the top parts of the grey regions. In addition, \textbf{CRS-Data} compares favorably to other methods in various DGPs and heterogeneity parameters.  \par

To conclude, the simulation results in this section show that the data-driven procedures to combine clusters perform well under various settings.

\begin{figure}[!ht]
	\centering
	\caption{Rejection rate curves for the simulations at significance level $\alpha=0.1$.}
	\label{fig:sim-010-new}
	\includegraphics[scale=1]{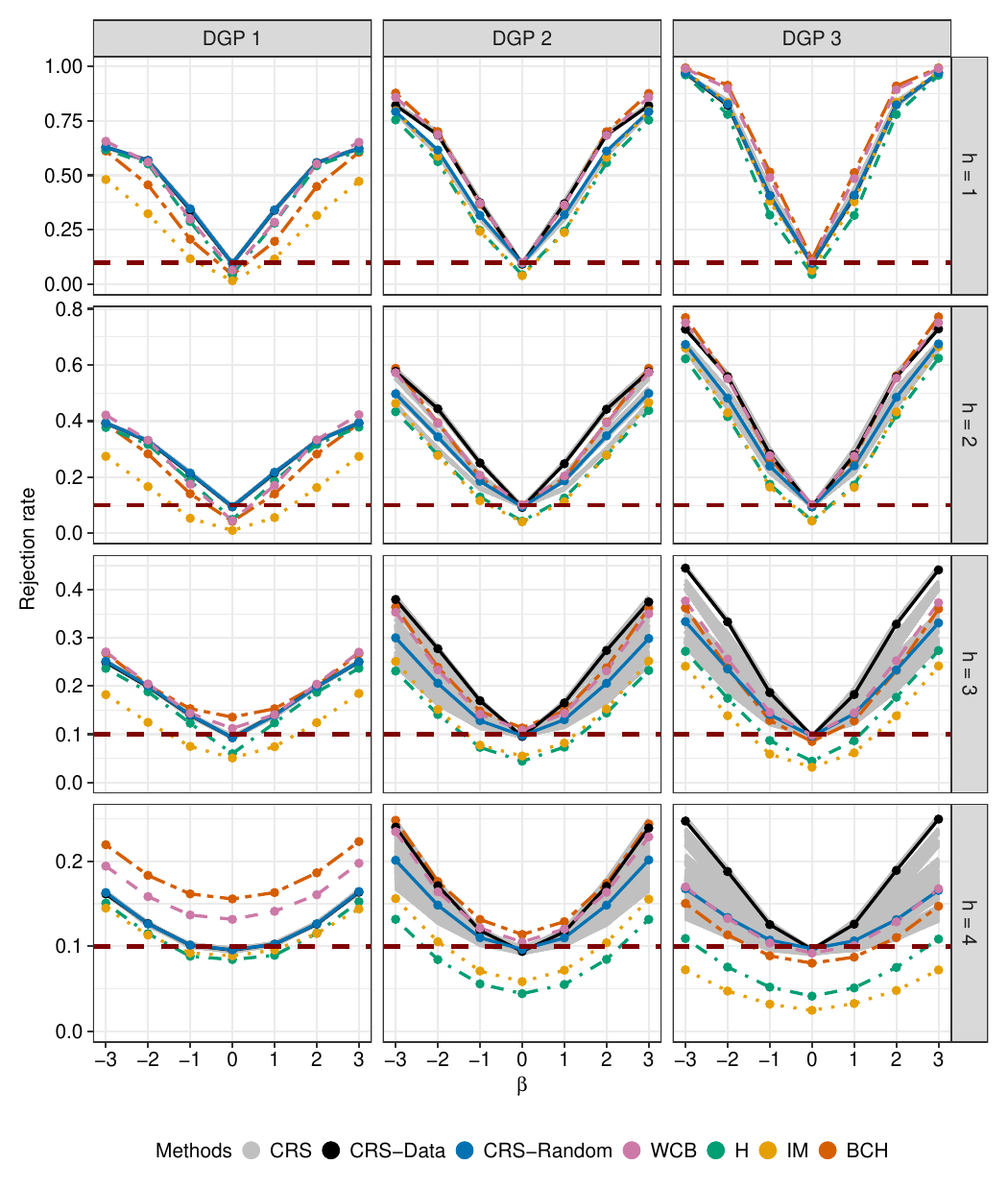}
\end{figure}

\section{Empirical application}	\label{sec:6}

\subsection{Setup}
\citet{dinceccokatz2016ej} study the impact of fiscal and administrative powers of states on economic performance, using data from 11 European countries. Table 3 of \citet{dinceccokatz2016ej} regresses real per capita GDP growth on two binary variables representing fiscal centralization and limited government, together with other variables and fixed effects. Their results are clustered at the country level.

The baseline regression they consider is:
\begin{equation}
	\label{eq:emp-ej-1}
	Y_{t,j} = \beta_0 + \beta_1 C_{t,j} + \beta_2 L_{t,j} + \mu_j  + U_{t,j},
\end{equation}
where $j$ indices country, $t$ indices year, $Y_{t,j}$ is the (logarithm) annual growth rate of real per capita GDP in country $j$ between years $t-1$ and $t$, $C_{t,j} \in \{0, 1\}$ equals 1 for fiscal centralization in country $j$ and year $t$, $L_{t,j}  \in \{0, 1\}$ equals 1 for limited government in country $j$ and year $t$, and $\mu_j$ is the country fixed-effects. The goal is to conduct inference on $\widehat \beta_1$ and $\widehat \beta_2$.

I focus on columns (1) to (3) of Table 3 in  \citet{dinceccokatz2016ej} in order to focus on the issue that there exist clusters in which the binary variables $C_{t,j}$ or $L_{t,j}$ do not have variation over time. The specifications for the three columns are as follows.
\begin{description}
	\item[Column (1).] As in \eqref{eq:emp-ej-1}.
	\item[Column (2).] As in column (1) and with year fixed effects.
	\item[Column (3).] As in column (2) and with country-specfic time trends.
\end{description}

In this situation of having only 11 clusters, CRS can be used to conduct valid inference. However, there are clusters with no variation in the main variables of interest. There are three countries with $C_{t,j} = 1$ in all periods (Belgium, England, and Piedmont). There are two countries with $L_{t,j} = 1$ in all periods (Belgium and Denmark). As a result, researchers would have to combine clusters in order to ensure that all parameters can be identified in each combined cluster. In addition, the data is at the aggregate level in that there is one observation for each year in each cluster. Therefore, researchers would also have to combine clusters in the presence of time fixed effects. In this section, I form five groups of clusters across all specifications.

In order to study the performance of the data-driven procedures in Section \ref{sec:4}, this section performs a calibrated simulation exercise. I first estimate the model using the actual data. Then, I simulate data by using various distributions of the error terms. In each specification, let $\widehat\beta_{1}$ and $\widehat\beta_{2}$ be the estimates of $\beta_1$ and $\beta_2$ obtained from data respectively. I take them as the true values of $\beta_1$ and $\beta_2$ in the calibrated simulation procedure. In addition, I perform inference on the coefficients of the two binary variables fiscal centralization $C_{t,j}$ and limited government $L_{t,j}$ separately. Hence, in conducting calibrated simulation exercises for the coefficient on $C_{t,j}$, I set $\beta_1 = \widehat\beta_{1}$ under the null, $\beta_1 = \widehat\beta_{1} + \Delta$ under the alternative, and keep $\beta_2 = \widehat\beta_{2}$. Similarly, for the coefficient on $L_{t,j}$, I set $\beta_2 = \widehat\beta_{2}$ under the null, $\beta_2 = \widehat\beta_{2} + \Delta$ under the alternative, and keep $\beta_1 = \widehat\beta_{1}$. The significance level is $\alpha=0.1$ for all simulations. There are 9,660 different ways to combine the clusters in this calibrated simulation exercise. I use 2,000 Monte Carlo simulations for each column and specification. I consider $\Delta \in \{-3, -2, \ldots, 3\}$.  \par 

In the calibrated simulation exercise, I first use the residuals $\{U_{t,j}\}$ to fit an AR(1) model 
\[
	U_{t,j} = \rho_j U_{t,j-1} + \epsilon_{t,j},
\]
where $\epsilon_{t,j} \sim N(0, \nu_j^2)$ for each country $j \in \cJ$. Using the estimated parameters $\{\widehat\rho_j\}_{j \in \cJ}$ and $\{\widehat\nu_j\}_{j \in \cJ}$, I consider the following specifications of the error terms in conducting the calibrated simulation exercise.
\begin{description}
\item[Specification (1).] Use the original $\widehat\nu_j$ for each country $j \in \cJ$.
\item[Specification (2).] Same as specification 1, but replace $\widehat\nu_j$ by $10\widehat\nu_j$ for $j \leq 4$.
\item[Specification (3).] Same as specification 1, but replace $\widehat\nu_j$ by $10\widehat\nu_j$ for $j \leq 4$ and $\widehat\nu_j$ by $5\widehat\nu_j$ for $5 \leq j \leq 8$.
\end{description}
More details of the calibrated simulation exercise can be found in Appendix \ref{app:d}. 

\subsection{Results}
Figures \ref{fig:sim-centralized} and \ref{fig:sim-limited} report the results for inference on the coefficients of the variables $C_{t,j}$ and $L_{t,j}$ respectively. Each figure reports the rejection rates against $\Delta$, where $\Delta$ was defined two paragraphs ago. $\Delta = 0$ corresponds to the results under the null. Based on the findings in Section \ref{sec:5} that CRS has favorable performance across various designs, I focus on examining the performance of the data-driven approach here. See Table \ref{tab:5.abbrev} for the definitions of \textbf{CRS-Data} and \textbf{CRS-Random}. Similar to Section \ref{sec:5}, the grey regions are formed by the 9,660 rejection rate curves that are based on the different ways of forming five groups out of 11 clusters. As in the previous section, the data-driven procedure \textbf{CRS-Data} uses the local alternative parameter $\delta = 2\sqrt{q T}$ under $\Delta \geq 0$ and $\delta = -2\sqrt{q T}$ under $\Delta < 0$, where $T$ is the number of periods in the calibrated simulation exercise. \par

\begin{figure}[!ht]
	\centering
	\caption{Rejection rate curves for inference on the coefficient for fiscal centralization $C_{t,j}$.}
	\label{fig:sim-centralized}
	\includegraphics[scale=1]{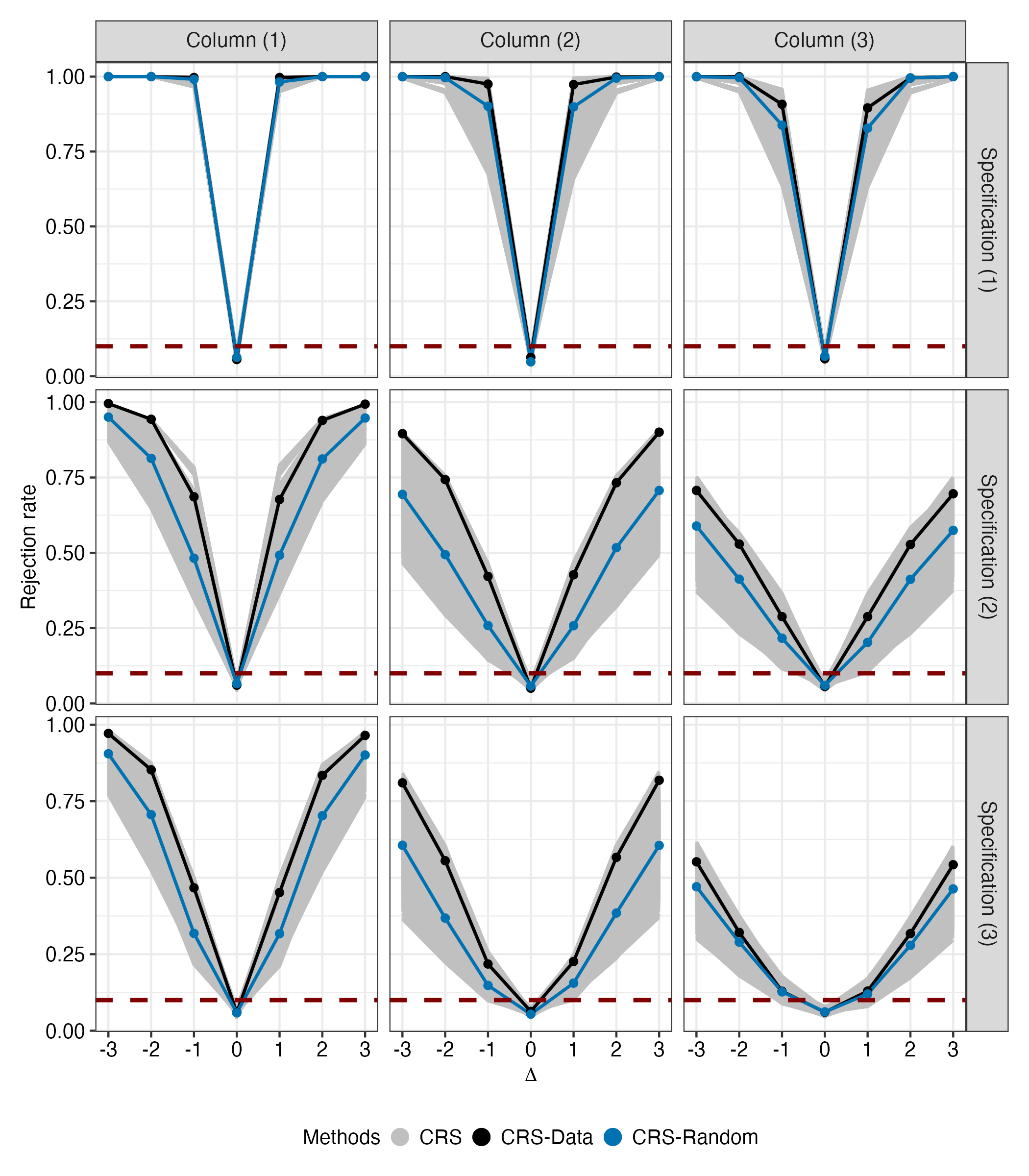}
\end{figure}

\begin{figure}[!ht]
	\centering
	\caption{Rejection rate curves for inference on the coefficient for limited government $L_{t,j}$.}
	\label{fig:sim-limited}
	\includegraphics[scale=1]{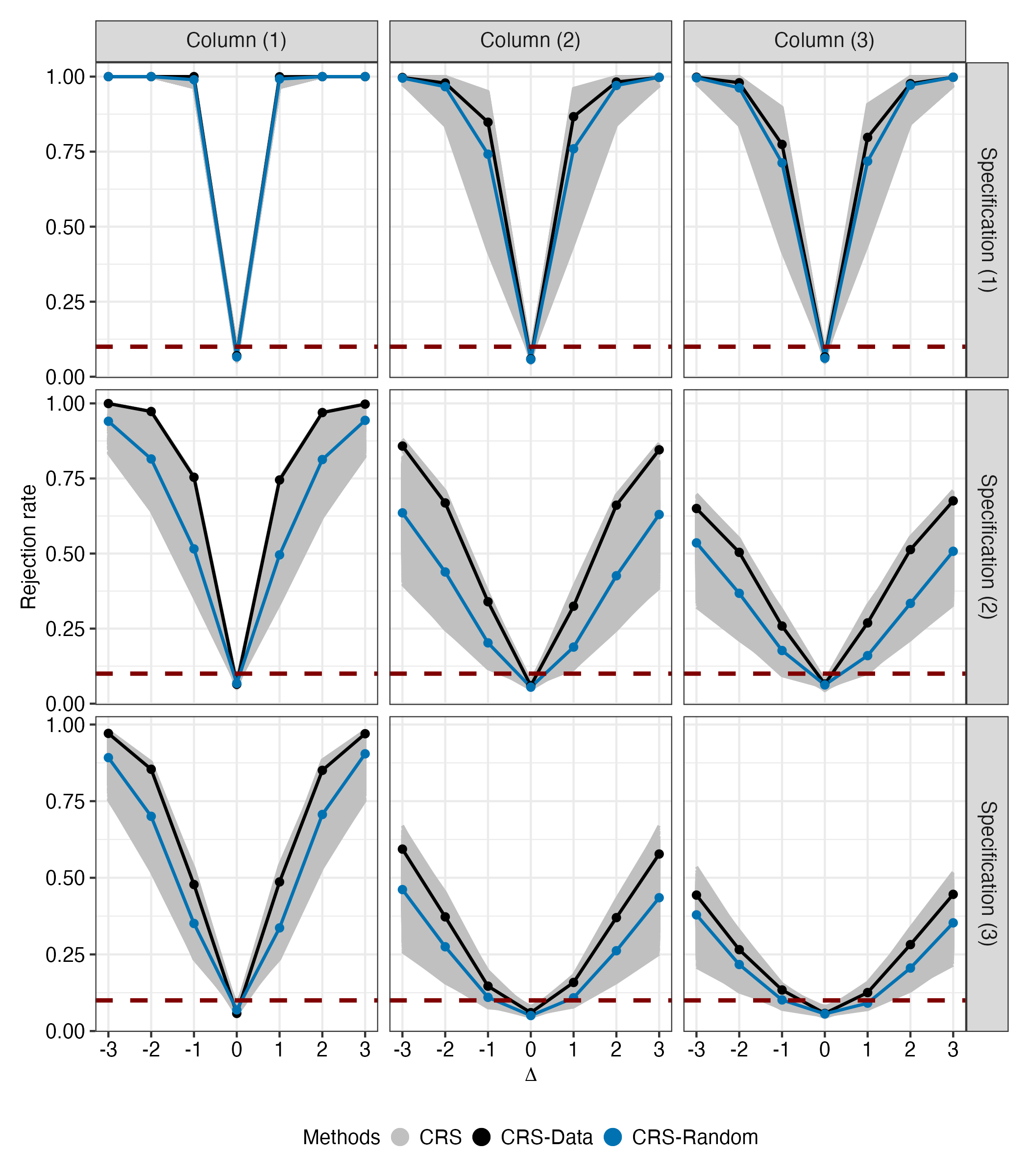}
\end{figure}

The simulation results show that \textbf{CRS-Data} controls size and chooses powerful grouping of clusters in various designs. It also continues to perform better than choosing a random grouping of clusters.

\newpage

\section{Conclusion} \label{sec:7}
The approximate randomization test in \citet{canayetal2017ecta} imposes weak assumptions on clusters and can be used to conduct valid inference when there is a small number of clusters. The test requires researchers to perform cluster-by-cluster regressions. When the target parameters cannot be identified within each cluster, researchers would have to combine clusters to perform the test.

In this paper, I first analyzed the local asymptotic power of CRS. Using the local asymptotic power as a criterion, I developed computationally efficient algorithms that can be used to guide how to combine clusters when there is an identification issue within cluster. Monte Carlo simulations and an empirical application show that the data-driven procedure performs well.

\newpage

\appendix
\numberwithin{equation}{section}
\counterwithin{figure}{section}

\bigskip  \bigskip	\bigskip
\begin{center}

{\bf \Large Appendix} 
\end{center}

The appendix contains all the proofs for the results in the main text and details of the calibrated simulation exercise in Section \ref{sec:6}.

\section{Supplemental lemma}
The lemma below is similar to Lemma \ref{lem:3.3.3}, except that it considers the equality of the two terms.

\begin{lem} \label{lem:b.1}
For any $g, h \in \bG$, let $\cJ_{\mathrm{same}}(g, h)$ and $\cJ_{\mathrm{diff}}(g, h)$ be as defined in Lemma \ref{lem:3.3.3}. Then,
\begin{align*}
		\{T_n(h) = T_n(g)\}
		& = 
		\left\{ \sum_{j \in \cJ_{\mathrm{same}}(g, h)} h_j \ws_{n,j} = 0 \right\} 	\\
		& \quad\quad \bigcup \left\{ \sum_{j \in \cJ_{\mathrm{same}}(g, h)} h_j \ws_{n,j} \neq 0, \sum_{j \in \cJ_{\mathrm{diff}}(g, h)} h_j \ws_{n,j} = 0 \right\} .
	\end{align*}	
\end{lem}

\begin{proof}[\textit{\textbf{Proof of Lemma \ref{lem:b.1}}}] 
To begin with, write $\wtt_n(g) \equiv \sum^q_{j=1} g_j \ws_{n,j}$ so that $T_n(g) = |\wtt_n(g)|$ for any $g \in \bG$. Note that the event $\{T_n (h) = T_n(g)\}$ can be written as the following nine unions of events depending on the signs of $\wtt_n(h)$ and $\wtt_n(g)$:
\begin{align}
	\{T_n (h) = T_n(g)\}
	& = \{|\wtt_n(h)| = |\wtt_n(g)|\} \notag \\
	& = \{\wtt_n(h) = \wtt_n(g),  \wtt_n(h) > 0,  \wtt_n(g) > 0\} 	\label{eq:lemb.1-1}	\\
	& \ \ \ \ \  \cup \{\wtt_n(h) = \wtt_n(g),  \wtt_n(h) = 0,  \wtt_n(g) > 0\} 	\label{eq:lemb.1-2}	\\
	& \ \ \ \ \  \cup \{-\wtt_n(h) = \wtt_n(g),  \wtt_n(h) < 0,  \wtt_n(g) > 0\} 	\label{eq:lemb.1-3}	\\
	& \ \ \ \ \  \cup \{\wtt_n(h) = \wtt_n(g),  \wtt_n(h) > 0,  \wtt_n(g) = 0\} 	\label{eq:lemb.1-4}	\\
	& \ \ \ \ \  \cup \{\wtt_n(h) = \wtt_n(g),  \wtt_n(h) = 0,  \wtt_n(g) = 0\} 	\label{eq:lemb.1-5}	\\
	& \ \ \ \ \  \cup \{-\wtt_n(h) = \wtt_n(g),  \wtt_n(h) < 0,  \wtt_n(g) = 0\} 	\label{eq:lemb.1-6}	\\
	& \ \ \ \ \  \cup \{\wtt_n(h) = -\wtt_n(g),  \wtt_n(h) > 0,  \wtt_n(g) < 0\} 	\label{eq:lemb.1-7}	\\
	& \ \ \ \ \  \cup \{\wtt_n(h) = -\wtt_n(g),  \wtt_n(h) = 0,  \wtt_n(g) < 0\} 	\label{eq:lemb.1-8}	\\
	& \ \ \ \ \  \cup \{-\wtt_n(h) = -\wtt_n(g),  \wtt_n(h) < 0,  \wtt_n(g) < 0\}. 	\label{eq:lemb.1-9}	
\end{align}
The events \eqref{eq:lemb.1-2}, \eqref{eq:lemb.1-4}, \eqref{eq:lemb.1-6} and \eqref{eq:lemb.1-8} are empty because each of these four events requires $\widetilde T_n(g) = 0$ and $\widetilde T_n(g) \neq 0$ at the same time.  \par 

Using the definitions of the sets $\cJ_{\textrm{same}}(g, h)$ and $\cJ_{\textrm{diff}}(g, h)$, $\wtt_n(h)$ can be written as
\begin{align}
	\label{eq:a-tnh}
	\wtt_n(h)
	= \sum^q_{j=1} h_j \ws_{n,j} 
	= \sum_{j \in \cJ_{\textrm{same}}(g, h)} h_j \ws_{n,j}
	+ \sum_{j \in \cJ_{\textrm{diff}}(g, h)} h_j \ws_{n,j},
\end{align}
and $\wtt_n(g)$ can be written in terms of the sign changes in $h$ as
\begin{align}
	\wtt_n(g)
	& = \sum^q_{j=1} g_j \ws_{n,j} 		\notag	\\
	& = \sum_{j \in \cJ_{\textrm{same}}(g, h)} g_j \ws_{n,j}
	+ \sum_{j \in \cJ_{\textrm{diff}}(g, h)} g_j \ws_{n,j}		\notag	\\
	& = \sum_{j \in \cJ_{\textrm{same}}(g, h)} h_j \ws_{n,j}
	- \sum_{j \in \cJ_{\textrm{diff}}(g, h)} h_j \ws_{n,j}.		\label{eq:a-tng}
\end{align}

The next step is to simplify the remaining events.
\begin{itemize}
	\item \eqref{eq:lemb.1-1} can be rewritten as follows:
	\begin{align}
		& \hspace{-20pt} \{\wtt_n(h) = \wtt_n(g),  \wtt_n(h) > 0,  \wtt_n(g) > 0\} \notag \\
		& =  \{\wtt_n(h) =  \wtt_n(g), \ \wtt_n(g) > 0 \} \notag	\\
	& = \left\{
		\sum_{j=1}^q h_j \ws_{n,j} = \sum_{j=1}^q g_j \ws_{n,j} , \ 
		\sum_{j=1}^q g_j \ws_{n,j} > 0
	\right\}	\notag \\
	& = \Bigg\{
		\sum_{j \in \cJ_{\textrm{same}}(g, h)} h_j \ws_{n,j} + 
		\sum_{j \in \cJ_{\textrm{diff}}(g, h)} h_j \ws_{n,j} 
		= 
		\sum_{j \in \cJ_{\textrm{same}}(g, h)} h_j \ws_{n,j} -
		\sum_{j \in \cJ_{\textrm{diff}}(g, h)} h_j \ws_{n,j},
	\notag \\
	& \qquad
		\sum_{j \in \cJ_{\textrm{same}}(g, h)} h_j \ws_{n,j} -
		\sum_{j \in \cJ_{\textrm{diff}}(g, h)} h_j \ws_{n,j}  > 0
	\Bigg\}	\notag \\
	& = \left\{
		\sum_{j \in \cJ_{\textrm{diff}}(g, h)} h_j \ws_{n,j} = 0,
		\ 
		\sum_{j \in \cJ_{\textrm{same}}(g, h)} h_j \ws_{n,j} >
		\sum_{j \in \cJ_{\textrm{diff}}(g, h)} h_j \ws_{n,j} 
	\right\}	\notag \\
	& = \left\{
		\sum_{j \in \cJ_{\textrm{diff}}(g, h)} h_j \ws_{n,j} = 0,
		\ 
		\sum_{j \in \cJ_{\textrm{same}}(g, h)} h_j \ws_{n,j} > 0
	\right\}. \label{eq:lemb.1-1b}
	\end{align}
	\item \eqref{eq:lemb.1-3} can be rewritten as follows:
	\begin{align}
		& \hspace{-20pt} \{-\wtt_n(h) = \wtt_n(g),  \wtt_n(h) < 0,  \wtt_n(g) > 0\} \notag \\
		& =  \{-\wtt_n(h) =  \wtt_n(g), \ \wtt_n(g) > 0 \} \notag	\\
	& = \left\{
		-\sum_{j=1}^q h_j \ws_{n,j} = \sum_{j=1}^q g_j \ws_{n,j} , \ 
		\sum_{j=1}^q g_j \ws_{n,j} > 0
	\right\}	\notag \\
	& = \Bigg\{
		-\sum_{j \in \cJ_{\textrm{same}}(g, h)} h_j \ws_{n,j} 
		-\sum_{j \in \cJ_{\textrm{diff}}(g, h)} h_j \ws_{n,j} 
		= 
		\sum_{j \in \cJ_{\textrm{same}}(g, h)} h_j \ws_{n,j} -
		\sum_{j \in \cJ_{\textrm{diff}}(g, h)} h_j \ws_{n,j},
	\notag \\
	& \qquad
		\sum_{j \in \cJ_{\textrm{same}}(g, h)} h_j \ws_{n,j} -
		\sum_{j \in \cJ_{\textrm{diff}}(g, h)} h_j \ws_{n,j}  > 0
	\Bigg\}	\notag \\
	& = \left\{
		\sum_{j \in \cJ_{\textrm{same}}(g, h)} h_j \ws_{n,j} = 0,
		\ 
		\sum_{j \in \cJ_{\textrm{same}}(g, h)} h_j \ws_{n,j} >
		\sum_{j \in \cJ_{\textrm{diff}}(g, h)} h_j \ws_{n,j} 
		\right\}			\notag \\
	& = \left\{
		\sum_{j \in \cJ_{\textrm{same}}(g, h)} h_j \ws_{n,j} = 0,
		\ 
		\sum_{j \in \cJ_{\textrm{diff}}(g, h)} h_j \ws_{n,j} < 0
	\right\}. 	\label{eq:lemb.1-3b}
	\end{align}
	\item \eqref{eq:lemb.1-5} can be rewritten as follows:
	\begin{align}
		& \hspace{-50pt} \{\wtt_n(h) = \wtt_n(g),  \wtt_n(h) = 0,  \wtt_n(g) = 0\} \notag  \\
		& =  \{\wtt_n(h) =  0, \ \wtt_n(g) = 0 \} \notag	\\
	& = \left\{
		\sum_{j=1}^q h_j \ws_{n,j} = 0 , \ 
		\sum_{j=1}^q g_j \ws_{n,j} = 0
	\right\}	\notag \\
	& = \Bigg\{
		\sum_{j \in \cJ_{\textrm{same}}(g, h)} h_j \ws_{n,j} 
		+ \sum_{j \in \cJ_{\textrm{diff}}(g, h)} h_j \ws_{n,j} 
		=  0
		, \notag \\
	& \qquad
		\sum_{j \in \cJ_{\textrm{same}}(g, h)} h_j \ws_{n,j} -
		\sum_{j \in \cJ_{\textrm{diff}}(g, h)} h_j \ws_{n,j}  = 0
	\Bigg\}	\notag \\
	& = \left\{
		\sum_{j \in \cJ_{\textrm{same}}(g, h)} h_j \ws_{n,j} = 0,
		\ 
		\sum_{j \in \cJ_{\textrm{diff}}(g, h)} h_j \ws_{n,j} = 0
	\right\}. 	\label{eq:lemb.1-5b}
	\end{align}
	\item \eqref{eq:lemb.1-7} can be rewritten as follows:
	\begin{align}
		& \hspace{-20pt} \{\wtt_n(h) = -\wtt_n(g),  \wtt_n(h) > 0,  \wtt_n(g) < 0\} \notag  \\
		& =  \{\wtt_n(h) =  -\wtt_n(g), \ \wtt_n(g) < 0 \} \notag	\\
	& = \left\{
		\sum_{j=1}^q h_j \ws_{n,j} = -\sum_{j=1}^q g_j \ws_{n,j} , \ 
		\sum_{j=1}^q g_j \ws_{n,j} < 0
	\right\}	\notag \\
	& = \Bigg\{
		\sum_{j \in \cJ_{\textrm{same}}(g, h)} h_j \ws_{n,j} 
		+ \sum_{j \in \cJ_{\textrm{diff}}(g, h)} h_j \ws_{n,j} 
		= 
		- \sum_{j \in \cJ_{\textrm{same}}(g, h)} h_j \ws_{n,j} 
		+
		\sum_{j \in \cJ_{\textrm{diff}}(g, h)} h_j \ws_{n,j},
	\notag \\
	& \qquad
		\sum_{j \in \cJ_{\textrm{same}}(g, h)} h_j \ws_{n,j} -
		\sum_{j \in \cJ_{\textrm{diff}}(g, h)} h_j \ws_{n,j}  < 0
	\Bigg\}	\notag \\
	& = \left\{
		\sum_{j \in \cJ_{\textrm{same}}(g, h)} h_j \ws_{n,j} = 0,
		\ 
		\sum_{j \in \cJ_{\textrm{same}}(g, h)} h_j \ws_{n,j} <
		\sum_{j \in \cJ_{\textrm{diff}}(g, h)} h_j \ws_{n,j} 
		\right\}			\notag \\
	& = \left\{
		\sum_{j \in \cJ_{\textrm{same}}(g, h)} h_j \ws_{n,j} = 0,
		\ 
		\sum_{j \in \cJ_{\textrm{diff}}(g, h)} h_j \ws_{n,j} > 0
	\right\}. 	
	\label{eq:lemb.1-7b}
	\end{align}
	\item \eqref{eq:lemb.1-9} can be rewritten as follows:
	\begin{align}
		& \hspace{-20pt} \{-\wtt_n(h) = -\wtt_n(g),  \wtt_n(h) < 0,  \wtt_n(g) < 0\} \notag  \\
		& =  \{-\wtt_n(h) =  -\wtt_n(g), \ \wtt_n(g) < 0 \} \notag	\\
	& = \left\{
		-\sum_{j=1}^q h_j \ws_{n,j} = -\sum_{j=1}^q g_j \ws_{n,j} , \ 
		\sum_{j=1}^q g_j \ws_{n,j} < 0
	\right\}	\notag \\
	& = \Bigg\{
		- \sum_{j \in \cJ_{\textrm{same}}(g, h)} h_j \ws_{n,j} 
		- \sum_{j \in \cJ_{\textrm{diff}}(g, h)} h_j \ws_{n,j} 
		= 
		- \sum_{j \in \cJ_{\textrm{same}}(g, h)} h_j \ws_{n,j} 
		+
		\sum_{j \in \cJ_{\textrm{diff}}(g, h)} h_j \ws_{n,j},
	\notag \\
	& \qquad
		\sum_{j \in \cJ_{\textrm{same}}(g, h)} h_j \ws_{n,j} -
		\sum_{j \in \cJ_{\textrm{diff}}(g, h)} h_j \ws_{n,j}  < 0
	\Bigg\}	\notag \\
	& = \left\{
		\sum_{j \in \cJ_{\textrm{diff}}(g, h)} h_j \ws_{n,j} = 0,
		\ 
		\sum_{j \in \cJ_{\textrm{same}}(g, h)} h_j \ws_{n,j} <
		\sum_{j \in \cJ_{\textrm{diff}}(g, h)} h_j \ws_{n,j} 
		\right\}			\notag \\
	& = \left\{
		\sum_{j \in \cJ_{\textrm{diff}}(g, h)} h_j \ws_{n,j} = 0,
		\ 
		\sum_{j \in \cJ_{\textrm{same}}(g, h)} h_j \ws_{n,j} < 0
	\right\}. 	\label{eq:lemb.1-9b}
	\end{align}
\end{itemize}

Let $A \equiv \sum_{j \in \cJ_{\textrm{same}}(g, h)} h_j \ws_{n,j} $ and $B \equiv \sum_{j \in \cJ_{\textrm{diff}}(g, h)} h_j \ws_{n,j} $. Then, the union of events \eqref{eq:lemb.1-1b} to \eqref{eq:lemb.1-9b} become
\begin{align*}
	& \hspace{-20pt} \{A > 0, B = 0\} \cup  \{A = 0,  B < 0\} \cup	 \{A = 0,  B = 0\} \cup
		 \{A = 0,  B > 0\} \cup \{A < 0, B = 0\}		\\
	& =  ( \{A = 0,  B < 0\} \cup	 \{A = 0,  B = 0\}  \cup \{A = 0, B > 0\})		\\
	& \quad\quad \cup ( \{B = 0,  A >0\} \cup	 \{B = 0,  A < 0\}  ) \\
	& = \{A = 0\} \cup \{A \neq 0, B = 0\},
\end{align*}
which yields the desired result. \end{proof}

The following lemma describes the limiting distribution under the local alternative. This result is used in multiple proofs so it is stated as a lemma below.
\begin{lem} \label{lem:b.2}
Let Assumption \ref{assu:2.2} hold and let $\delta \in \bR$ be the local alternative parameter as in \eqref{eq:3.1}.  Then,

\[
	\begin{pmatrix}
		\ws_{n,1} \\
		\vdots \\
		\ws_{n,q}
	\end{pmatrix}
	\Dt 
	\begin{pmatrix}
		Z_1 \\
		\vdots \\
		Z_q
	\end{pmatrix}
	+ 
	\begin{pmatrix}
		\xi_1 \\
		\vdots \\
		\xi_q
	\end{pmatrix}
	\delta,
\]
where $Z_j \equiv c'S_j \sim N(0, \sigma_j^2)$ and $\sigma_j^2 \equiv c'\Sigma_j c$ for all $j \in \cJ$. In addition, $Z_j \indep Z_k$ for any $j \neq k$.
\end{lem}

\begin{proof}[\textit{\textbf{Proof of Lemma \ref{lem:b.2}}}] 
Under the local alternative that $c'\beta = \lambda + \frac{\delta}{\sqrt{n}}$, the following holds for each $j \in \cJ$:
\begin{align*}
	\ws_{n,j}
	 = \sqrt{n_j}(c'\widehat\beta_{n,j} - \lambda)	
	 = \sqrt{n_j}\left(c'\widehat\beta_{n,j} - c'\beta + \frac{\delta}{\sqrt{n}}\right)	
	 = \sqrt{n_j}c'(\widehat\beta_{n,j} - \beta) + \frac{\sqrt{n_j}\delta}{\sqrt{n}} . 
\end{align*}
Therefore,
\begin{align*}
	\begin{pmatrix}
		\ws_{n,1} \\
		\vdots \\
		\ws_{n,q}
	\end{pmatrix}
	& = 
	\begin{pmatrix}
		\sqrt{n_1}c'(\widehat\beta_{n,1} - \beta) 	\\
		\vdots \\
		\sqrt{n_q}c'(\widehat\beta_{n,q} - \beta)
	\end{pmatrix}
	+
	\begin{pmatrix}
		\frac{\sqrt{n_1}}{\sqrt{n}} 	\\
		\vdots \\
		\frac{\sqrt{n_q}}{\sqrt{n}} 
	\end{pmatrix}
	\delta \\
	& \Dt 
	\begin{pmatrix}
		c'S_1 \\
		\vdots \\
		c'S_q
	\end{pmatrix}
	+ 
	\begin{pmatrix}
		\xi_1 \\
		\vdots \\
		\xi_q
	\end{pmatrix}
	\delta \\
	& =
	\begin{pmatrix}
		Z_1 \\
		\vdots \\
		Z_q
	\end{pmatrix}
	+ 
	\begin{pmatrix}
		\xi_1 \\
		\vdots \\
		\xi_q
	\end{pmatrix}
	\delta,
\end{align*}
by Slutsky's theorem, Assumption \ref{assu:2.2} and by defining $Z_j$ as in the statement of the current lemma. In addition, $Z_j \indep Z_k$ for any $j \neq k$ follows from $S_j \indep S_k$ for any $j \neq k$ in Assumption \ref{assu:2.2}.
\end{proof}

\section{Proofs to the main text} \label{app:A}

\begin{proof}[\textit{\textbf{Proof of Lemma \ref{lem:3.3.3}}}]
To begin with, write $\wtt_n(g) \equiv \sum^q_{j=1} g_j \ws_{n,j}$ so that $T_n(g) = |\wtt_n(g)|$.
Then, the event $\{T_n(h) > T_n(g)\}$ can be written as the following four unions of  events depending on the signs of $\wtt_n(h)$ and $\wtt_n(g)$:
\begin{align}
	\{T_n (h) > T_n(g)\}
	& = \{|\wtt_n(h)| > |\wtt_n(g)|\} \notag \\
	& = \{\wtt_n(h) > \wtt_n(g),  \wtt_n(h) \geq 0,  \wtt_n(g) \geq 0\} 	\label{eq:lem3.3.3-1}	\\
	& \ \ \ \ \  \cup \{\wtt_n(h) > -\wtt_n(g),  \wtt_n(h) \geq 0,  \wtt_n(g) < 0\} 	\label{eq:lem3.3.3-2}	\\
	& \ \ \ \ \  \cup \{-\wtt_n(h) > \wtt_n(g),  \wtt_n(h) < 0,  \wtt_n(g) \geq 0\} 	\label{eq:lem3.3.3-3}	\\
	& \ \ \ \ \  \cup \{-\wtt_n(h) > -\wtt_n(g),  \wtt_n(h) < 0,  \wtt_n(g) < 0\}. 	\label{eq:lem3.3.3-4}
\end{align}

The next step is to simplify each of the four events in \eqref{eq:lem3.3.3-1} to \eqref{eq:lem3.3.3-4}. Note that each of the four events above contains the intersection of three inequalities. The second inequality is implied by the first and third inequalities in each event. In addition, $\wtt_n(h)$ and $\wtt_n(g)$ can be written as in \eqref{eq:a-tnh} and \eqref{eq:a-tng}.

\begin{itemize}
	\item \eqref{eq:lem3.3.3-1} can be rewritten as follows:
\begin{align}
	& \hspace{-20pt} \{\wtt_n(h) > \wtt_n(g),  \wtt_n(h) \geq 0,  \wtt_n(g) \geq 0\}	\notag \\ 
	& = 
	\{\wtt_n(h) > \wtt_n(g), \ \wtt_n(g) \geq 0 \} \notag	\\
	& = \left\{
		\sum_{j=1}^q h_j \ws_{n,j} > \sum_{j=1}^q g_j \ws_{n,j} , \ 
		\sum_{j=1}^q g_j \ws_{n,j} \geq 0
	\right\}	\notag \\
	& = \Bigg\{
		\sum_{j \in \cJ_{\textrm{same}}(g, h)} h_j \ws_{n,j} + 
		\sum_{j \in \cJ_{\textrm{diff}}(g, h)} h_j \ws_{n,j} 
		> 
		\sum_{j \in \cJ_{\textrm{same}}(g, h)} h_j \ws_{n,j} -
		\sum_{j \in \cJ_{\textrm{diff}}(g, h)} h_j \ws_{n,j},
	\notag \\
	& \qquad
		\sum_{j \in \cJ_{\textrm{same}}(g, h)} h_j \ws_{n,j} -
		\sum_{j \in \cJ_{\textrm{diff}}(g, h)} h_j \ws_{n,j}  \geq 0
	\Bigg\}	\notag \\
	& = \left\{
		\sum_{j \in \cJ_{\textrm{diff}}(g, h)} h_j \ws_{n,j} > 0,
		\ 
		\sum_{j \in \cJ_{\textrm{same}}(g, h)} h_j \ws_{n,j} \geq
		\sum_{j \in \cJ_{\textrm{diff}}(g, h)} h_j \ws_{n,j} 
	\right\}.	\label{eq:lem3.3.3-1b}
\end{align}
	\item \eqref{eq:lem3.3.3-2}  can be rewritten as follows:
\begin{align}
	& \hspace{-20pt} \{\wtt_n(h) > -\wtt_n(g),  \wtt_n(h) \geq 0,  \wtt_n(g) < 0\}	\notag \\ 
	& = 
	\{\wtt_n(h) > -\wtt_n(g), \ \wtt_n(g) < 0 \} \notag	\\
	& = \left\{
		\sum_{j=1}^q h_j \ws_{n,j} > -\sum_{j=1}^q g_j \ws_{n,j} , \ 
		\sum_{j=1}^q g_j \ws_{n,j} < 0
	\right\}	\notag \\
	& = \Bigg\{
		\sum_{j \in \cJ_{\textrm{same}}(g, h)} h_j \ws_{n,j} + 
		\sum_{j \in \cJ_{\textrm{diff}}(g, h)} h_j \ws_{n,j} 
		> 
		-\sum_{j \in \cJ_{\textrm{same}}(g, h)} h_j \ws_{n,j} +
		\sum_{j \in \cJ_{\textrm{diff}}(g, h)} h_j \ws_{n,j},
	\notag \\
	& \qquad
		\sum_{j \in \cJ_{\textrm{same}}(g, h)} h_j \ws_{n,j} -
		\sum_{j \in \cJ_{\textrm{diff}}(g, h)} h_j \ws_{n,j}  < 0
	\Bigg\}	\notag \\
	& = \left\{
		\sum_{j \in \cJ_{\textrm{same}}(g, h)} h_j \ws_{n,j} > 0,
		\ 
		\sum_{j \in \cJ_{\textrm{same}}(g, h)} h_j \ws_{n,j} <
		\sum_{j \in \cJ_{\textrm{diff}}(g, h)} h_j \ws_{n,j} 
	\right\}.	\label{eq:lem3.3.3-2b}
\end{align}

\item \eqref{eq:lem3.3.3-3} can be rewritten as follows:
\begin{align}
	& \hspace{-20pt} \{-\wtt_n(h) > \wtt_n(g),  \wtt_n(h) < 0,  \wtt_n(g) \geq 0\}	\notag \\ 
	& = 
	\{-\wtt_n(h) > \wtt_n(g), \ \wtt_n(g) \geq 0 \} \notag	\\
	& = \left\{
		- \sum_{j=1}^q h_j \ws_{n,j} > \sum_{j=1}^q g_j \ws_{n,j} , \ 
		\sum_{j=1}^q g_j \ws_{n,j} \geq 0
	\right\}	\notag \\
	& = \Bigg\{
		-\sum_{j \in \cJ_{\textrm{same}}(g, h)} h_j \ws_{n,j} - 
		\sum_{j \in \cJ_{\textrm{diff}}(g, h)} h_j \ws_{n,j} 
		> 
		\sum_{j \in \cJ_{\textrm{same}}(g, h)} h_j \ws_{n,j} -
		\sum_{j \in \cJ_{\textrm{diff}}(g, h)} h_j \ws_{n,j},
	\notag \\
	& \qquad
		\sum_{j \in \cJ_{\textrm{same}}(g, h)} h_j \ws_{n,j} -
		\sum_{j \in \cJ_{\textrm{diff}}(g, h)} h_j \ws_{n,j}  \geq 0
	\Bigg\}	\notag \\
	& = \left\{
		\sum_{j \in \cJ_{\textrm{same}}(g, h)} h_j \ws_{n,j} < 0,
		\ 
		\sum_{j \in \cJ_{\textrm{same}}(g, h)} h_j \ws_{n,j} \geq
		\sum_{j \in \cJ_{\textrm{diff}}(g, h)} h_j \ws_{n,j} 
	\right\}.	\label{eq:lem3.3.3-3b}
\end{align}

\item \eqref{eq:lem3.3.3-4} can be rewritten as follows:
\begin{align}
	& \hspace{-20pt} \{-\wtt_n(h) > -\wtt_n(g),  \wtt_n(h) < 0,  \wtt_n(g) < 0\}	\notag \\ 
	& = 
	\{-\wtt_n(h) > -\wtt_n(g), \ \wtt_n(g) < 0 \} \notag	\\
	& = \left\{
		- \sum_{j=1}^q h_j \ws_{n,j} > -\sum_{j=1}^q g_j \ws_{n,j} , \ 
		\sum_{j=1}^q g_j \ws_{n,j} < 0
	\right\}	\notag \\
	& = \Bigg\{
		-\sum_{j \in \cJ_{\textrm{same}}(g, h)} h_j \ws_{n,j} - 
		\sum_{j \in \cJ_{\textrm{diff}}(g, h)} h_j \ws_{n,j} 
		> 
		-\sum_{j \in \cJ_{\textrm{same}}(g, h)} h_j \ws_{n,j} +
		\sum_{j \in \cJ_{\textrm{diff}}(g, h)} h_j \ws_{n,j},
	\notag \\
	& \qquad
		\sum_{j \in \cJ_{\textrm{same}}(g, h)} h_j \ws_{n,j} -
		\sum_{j \in \cJ_{\textrm{diff}}(g, h)} h_j \ws_{n,j}  < 0
	\Bigg\}	\notag \\
	& = \left\{
		\sum_{j \in \cJ_{\textrm{diff}}(g, h)} h_j \ws_{n,j} < 0,
		\ 
		\sum_{j \in \cJ_{\textrm{same}}(g, h)} h_j \ws_{n,j} <
		\sum_{j \in \cJ_{\textrm{diff}}(g, h)} h_j \ws_{n,j} 
	\right\}.	\label{eq:lem3.3.3-4b}
\end{align}
\end{itemize}

Let $A \equiv \sum_{j \in \cJ_{\textrm{same}}(g, h)} h_j \ws_{n,j} $ and $B \equiv \sum_{j \in \cJ_{\textrm{diff}}(g, h)} h_j \ws_{n,j} $. Then, the union of events \eqref{eq:lem3.3.3-1b} to \eqref{eq:lem3.3.3-4b} become
\begin{align*}
	& \hspace{-20pt} \{B > 0, A \geq B\} \cup  \{A > 0, A < B\} \cup	 \{A < 0, A \geq B\} \cup
		 \{B < 0, A < B\}		\\
	& = (\{A > 0, B > 0, A \geq B\} \cup  \{A > 0, B > 0, A < B\} )	\\
	& \quad\quad \cup ( \{A < 0, B < 0, A \geq B\} \cup  \{A < 0, B < 0, A < B\}) \\
	& = \{A > 0, B > 0\} \cup \{A < 0, B < 0\},
\end{align*}
which yields the desired result. \end{proof}

\begin{proof}[\textbf{Proof of Lemma \ref{lem:3.3}}]
To begin with, define 
\[
	\cE_n(\cH_1) \equiv \{ \text{$T_n > T_n(h_2)$ for all $h_2 \in \bG_{U, -1} \backslash \cH_1$ and $T_n(h_1) > T_n$ for all $h_1 \in  \cH_1$}\},
\]
for any $\cH_1 \subseteq \bG_{U, -1}$. By definition, the event $\cE_n(\cH_1)$ can be written as the intersection of $\{T_n > T_n(g)\}$ for $g \in \bG_{U, -1} \backslash \cH_1$ and $\{T_n(g) > T_n\}$ for $g \in \cH_1$. Recall that $T_n = T_n(1_q)$, where $1_q \equiv (1, \ldots, 1)$ is the identity transformation. Using Lemma \ref{lem:3.3.3}, these events can be rewritten as
\begin{align}
	\{T_n > T_n(g)\}
	= \cA^{(1)}_n(g) \cup \cA^{(2)}_n(g),		\label{eq:pf-3.5.1}
\end{align} 
and
\begin{align}
	\{T_n(g)  > T_n\}
	= \cB^{(1)}_n(g) \cup \cB^{(2)}_n(g),		\label{eq:pf-3.5.2}
\end{align} 
where
\begin{align*}
	\cA^{(1)}_n(g)
	& \equiv 
	\left\{ \sum_{j \in \cJ_{\mathrm{same}}(g, 1_q)} \ws_{n,j} > 0, \sum_{j \in \cJ_{\mathrm{diff}}(g, 1_q)}  \ws_{n,j} > 0 \right\}, \\
	\cA^{(2)}_n(g)
	& \equiv 
	\left\{ \sum_{j \in \cJ_{\mathrm{same}}(g, 1_q)} \ws_{n,j} < 0, \sum_{j \in \cJ_{\mathrm{diff}}(g, 1_q)}  \ws_{n,j} < 0 \right\}, \\
	\cB^{(1)}_n(g)
	& \equiv 
	\left\{ \sum_{j \in \cJ_{\mathrm{same}}(g, 1_q)} \ws_{n,j} > 0, \sum_{j \in \cJ_{\mathrm{diff}}(g, 1_q)}  \ws_{n,j} < 0 \right\}, \\
	\cB^{(2)}_n(g)
	& \equiv 
	\left\{ \sum_{j \in \cJ_{\mathrm{same}}(g, 1_q)} \ws_{n,j} < 0, \sum_{j \in \cJ_{\mathrm{diff}}(g, 1_q)}  \ws_{n,j} > 0 \right\},
\end{align*}
for any $g \in \bG_{U, -1}$. \par 

Therefore, for any $\cH_1 \in \bH_{k-1}$, 
\begin{align}
	\cE_n(\cH_1) 
	& = \left(
		\bigcap_{h_2 \in \bG_{U, -1} \backslash \cH_1} 
		\{T_n > T_n(h_2)\}
	\right) \bigcap
	\left(
		\bigcap_{h_1 \in \cH_1} 
		\{T_n(h_1) > T_n \}
	\right)	 \notag \\
	& = \left(
		\bigcap_{h_2 \in \bG_{U, -1} \backslash \cH_1} 
		\{ \cA^{(1)}_n(h_2) \cup \cA^{(2)}_n(h_2)\}
	\right) \bigcap
	\left(
		\bigcap_{h_1 \in \cH_1} 
		\{ \cB^{(1)}_n(h_1) \cup \cB^{(2)}_n(h_1)\}	
	\right)	 \notag \\
	& =  \left(
			\bigcap_{h_2 \in \bG_{U, -1} \backslash \cH_1} 
			\{ \cF^{(1)}_n (h_2, \cH_1) \cup \cF^{(2)}_n(h_2, \cH_1)\}
		\right) \bigcap
	\left(
		\bigcap_{h_1 \in \cH_1} 
		\{ \cF^{(1)}_n(h_1, \cH_1) \cup \cF^{(2)}_n(h_1, \cH_1)\}	
	\right)  \notag \\
	& = \bigcup_{m_1 =1}^2
		 \bigcup_{m_2 =1}^2
		 \cdots
		  \bigcup_{m_{L} =1}^2
		\left(
		\bigcap^{L}_{l=1} 
		\cF^{(m_l)}_n(\og_l, \cH_1)
		\right)  \notag \\
	& = \bigcup_{m \in \bM}
		\left(
		\bigcap^{L}_{l=1} 
		\cF^{(m_l)}_n(\og_l, \cH_1)
		\right) ,	\label{eq:pf-3.5.0}
\end{align}
where the first equality follows from the definition of $\cE_n(\cH_1)$, the second equality follows from \eqref{eq:pf-3.5.1} and \eqref{eq:pf-3.5.2}, the third equality follows from relabelling
\[
	\cF^{(m_l)}_n(\og_l, \cH)
	\equiv \begin{cases}
		 \cA_n^{(m_l)}(\og_l) & , \ \og_l \in \bG_{U, -1} \backslash \cH_1	\\
		 \cB_n^{(m_l)}(\og_l) &, \ \og_l \in \cH_1
	\end{cases}
\]
for each $l = 1, \ldots, L$, the fourth equality uses the property on the intersection and union of a finite number of sets (note that $L$ is finite), and the last equality follows from labeling the unique elements in $\bG_{U, -1}$ as $\{\og_1, \ldots, \og_{L}\}$ as in the last paragraph of Section \ref{sec:3.1} and writing $\bM \equiv \{1, 2\}^L$ as in the statement of the proposition to simplify the $L$ unions. \par 

Next, for any $m, \widetilde m \in \bM$ with $m \neq \widetilde m$, there must exist at least one $v = 1, \ldots, L$ such that $m_v \neq \widetilde m_v$. Assume without loss of generality that $m_v = 1$ and $\widetilde m_v = 2$. Then, it must be the case that $\cF^{(m_v)}_n(\og_v, \cH_1) \cap \cF^{(\widetilde m_v)}_n(\og_v, \cH_1)  = \emptyset$ because $\cA^{(1)}_n(\og_v) \cap \cA^{(2)}_n(\og_v) = \emptyset$ if $\og_v \in \bG_{U, -1} \backslash \cH_1$ and $\cB^{(1)}_n(\og_v) \cap \cB^{(2)}_n(\og_v) = \emptyset$ if $\og_v \in \cH_1$. This follows because for a given $g \in \bG_{U, -1}$, $\cA^{(1)}_n(g)$ requires
\[
	\sum_{j \in \cJ} \ws_{n,j} > 0,
\]
whereas $\cA^{(2)}_n(g)$ requires
\[
	\sum_{j \in \cJ} \ws_{n,j} < 0.
\]
Similarly, for a given $g \in \bG_{U, -1}$, $\cB^{(1)}_n(g)$ requires
\[
	\sum_{j \in \cJ_{\mathrm{same}}(g, 1_q)} \ws_{n,j}
	-
	\sum_{j \in \cJ_{\mathrm{diff}}(g, 1_q)} \ws_{n,j}
	> 0,
\]
whereas $\cB^{(2)}_n(g)$ requires
\[
	\sum_{j \in \cJ_{\mathrm{same}}(g, 1_q)} \ws_{n,j}
	-
	\sum_{j \in \cJ_{\mathrm{diff}}(g, 1_q)} \ws_{n,j}
	< 0.
\]
As a result,
\begin{equation}
	\label{eq:pf-3.5.3}
		\left(
		\bigcap^{L}_{l=1} 
		\cF^{(m_l)}_n(\og_l, \cH_1)
		\right) 
	\bigcap
	\left(
		\bigcap^{L}_{l=1} 
		\cF^{(\widetilde m_l)}_n(\og_l, \cH_1)
		\right) 
	=
	\emptyset,
\end{equation}
whenever $m \neq \widetilde m$.

Using \eqref{eq:pf-3.5.0}, \eqref{eq:pf-3.5.3}, the inclusion-exclusion principle, and the fact that $|\bM|$ is finite, it follows that
\begin{align}
	\lim_{n\to\infty}  \bP_\delta[\cE_n(\cH_1)] 
	& = \sum_{m \in \bM} 
	\lim_{n\to\infty}  \bP_\delta\left[ 
		\left(
			\bigcap^{L}_{l=1}  \cF^{(m_l)}_n(\og_l, \cH_1)
		\right)
	\right] .	
	\label{eq:pf-3.5.4}
\end{align}

For each $l = 1, \ldots, L$, the event $\cF_n^{(m_l)}(\overline g_l, \cH_1)$ can be written as $\{F_l \ws_n > 0_{2\times 1}\}$, where $0_{2\times 1} \equiv (0, 0)'$, $\ws_n \equiv (\ws_{n,1}, \ldots, \ws_{n,q})'$ and $F_l$ is a $2 \times q$ matrix defined as follows.
\begin{itemize}
\item If $\overline g_l \in \bG_{U, -1} \backslash \cH_1$ and $m_l = 1$, then the $(i,j)$-entry of $F_l$ is defined as
\begin{equation}
	\label{eq:pf:lem3.3-a}
	\begin{cases}
		1	& , \ \text{if ($i=1$ and $j \in \cJ_{\mathrm{same}}(\overline g_l, 1_q)$) or ($i=2$ and $j \in \cJ_{\mathrm{diff}}(\overline g_l, 1_q)$)}	\\
		0 	& , \ \text{otherwise}
	\end{cases}.
\end{equation}
\item If $\overline g_l \in \bG_{U, -1} \backslash \cH_1$ and $m_l = 2$, then the $(i,j)$-entry of $F_l$ is defined as
\begin{equation}
	\label{eq:pf:lem3.3-b}
	\begin{cases}
		-1	& , \ \text{if ($i=1$ and $j \in \cJ_{\mathrm{same}}(\overline g_l, 1_q)$) or ($i=2$ and $j \in \cJ_{\mathrm{diff}}(\overline g_l, 1_q)$)}	\\
		0 	& , \ \text{otherwise}
	\end{cases}.
\end{equation}
\item If $\overline g_l \in \cH_1$ and $m_l = 1$, then the $(i,j)$-entry of $F_l$ is defined as
\begin{equation}
	\label{eq:pf:lem3.3-c}
	\begin{cases}
		1	& , \ \text{if $i=1$ and $j \in \cJ_{\mathrm{same}}(\overline g_l, 1_q)$}	\\
		-1	& , \ \text{if $i=2$ and $j \in \cJ_{\mathrm{diff}}(\overline g_l, 1_q)$}	\\
		0 	& , \ \text{otherwise}
	\end{cases}.
\end{equation}
\item If $\overline g_l \in  \cH_1$ and $m_l = 2$, then the $(i,j)$-entry of $F_l$ is defined as
\begin{equation}
	\label{eq:pf:lem3.3-d}
	\begin{cases}
		-1	& , \ \text{if $i=1$ and $j \in \cJ_{\mathrm{same}}(\overline g_l, 1_q)$}	\\
		1	& , \ \text{if $i=2$ and $j \in \cJ_{\mathrm{diff}}(\overline g_l, 1_q)$}	\\
		0 	& , \ \text{otherwise}
	\end{cases}.
\end{equation}
\end{itemize}
Let $F$ be the matrix that stacks $F_1, \ldots, F_L$ by row. Using Lemma \ref{lem:b.2} and the continuous mapping theorem, it follows that
\begin{align}
	F \ws_n \Dt F Z + F \xi \delta,
\end{align} 
where $Z \equiv (Z_1, \ldots, Z_q)'$ and $\xi \equiv (\xi_1, \ldots, \xi_q)'$ such that $Z_j \equiv c'S_j \sim N(0, \sigma_j^2)$ and $\sigma_j^2 \equiv c'\Sigma_j c$ for all $j \in \cJ$ as defined in Lemma \ref{lem:b.2}. In addition, $Z_j \indep Z_k$ for any $j \neq k$. Let
\[
	\cF^{(m_l)}(\og_l, \cH_1, \delta)
	\equiv \begin{cases}
		 \cA^{(m_l)}(\og_l, \delta) & , \ \og_l \in \bG_{U,-1} \backslash \cH_1	\\
		 \cB^{(m_l)}(\og_l, \delta) &, \ \og_l \in \cH_1
	\end{cases},
\]
for any $l = 1, \ldots, L$, and 
\begin{align*}
	\cA^{(1)}(g, \delta)
	& \equiv 
	\left\{
		\sum_{j \in \cJ_{\mathrm{same}}(g, 1_q)} Z_j 
			> -\sum_{j \in \cJ_{\mathrm{same}}(g, 1_q)} \xi_j \delta,
		\sum_{j \in \cJ_{\mathrm{diff}}(g, 1_q)}  Z_j 
			> -\sum_{j \in \cJ_{\mathrm{diff}}(g, 1_q)} \xi_j \delta 
	\right\}, \\
	\cA^{(2)}(g, \delta)
	& \equiv 
	\left\{ 
		\sum_{j \in \cJ_{\mathrm{same}}(g, 1_q)} Z_j 
			<  -\sum_{j \in \cJ_{\mathrm{same}}(g, 1_q)} \xi_j \delta, 
		\sum_{j \in \cJ_{\mathrm{diff}}(g, 1_q)}  Z_j 
			< -\sum_{j \in \cJ_{\mathrm{diff}}(g, 1_q)} \xi_j \delta
	\right\}, \\
	\cB^{(1)}(g, \delta)
	& \equiv 
	\left\{ 
		\sum_{j \in \cJ_{\mathrm{same}}(g, 1_q)} Z_j 
			>  -\sum_{j \in \cJ_{\mathrm{same}}(g, 1_q)} \xi_j \delta, 
		\sum_{j \in \cJ_{\mathrm{diff}}(g, 1_q)} Z_j 
			< -\sum_{j \in \cJ_{\mathrm{diff}}(g, 1_q)} \xi_j \delta 
	\right\}, \\
	\cB^{(2)}(g, \delta)
	& \equiv 
	\left\{ 
		\sum_{j \in \cJ_{\mathrm{same}}(g, 1_q)} Z_j 
			<  -\sum_{j \in \cJ_{\mathrm{same}}(g, 1_q)} \xi_j \delta, 
		\sum_{j \in \cJ_{\mathrm{diff}}(g, 1_q)} Z_j 
			> -\sum_{j \in \cJ_{\mathrm{diff}}(g, 1_q)} \xi_j \delta 
	\right\},
\end{align*}
for any $g \in \bG_{U, -1}$. \par 

Let $\mathfrak{B} \equiv \{(x_1, \ldots, x_{2L}) : x_l > 0 \text{ for each $l = 1, \ldots, 2L$}\}$, so that $\{F\widehat S_n > 0\}$ can be written as $\{F\widehat S_n \in \mathfrak{B}\}$. Denote $\partial \mathfrak{B}$ as the boundary of $\mathfrak{B}$. Note that $\partial \mathfrak{B} \subseteq \bigcup^{2L}_{l=1}  \{ x_l = 0\}$. Now, $\bP[\sum_{j \in \cJ_{\text{same}}(\og_l, 1_q)} Z_j = - \sum_{j \in \cJ_{\text{same}}(\og_l, 1_q)} \xi_j \delta] = 0$ for each $l = 1, \ldots, L$ by Assumption \ref{assu:3.3.4}. This follows by writing $w_0 = \sum_{j \in \cJ_{\mathrm{same}}(g, 1_q)} \xi_j \delta$, $w_j = 1$ for $j \in \cJ_{\text{same}}(\overline g_l, 1_q)$ and $w_j = 0$ for $j \notin \cJ_{\text{same}}(\overline g_l, 1_q)$ in Assumption \ref{assu:3.3.4}. For the same reason, $\bP[\sum_{j \in \cJ_{\text{diff}}(\og_l, 1_q)} Z_j = - \sum_{j \in \cJ_{\text{diff}}(\og_l, 1_q)} \xi_j \delta] = 0$ for each $l = 1, \ldots, L$. Therefore, $\bP[FZ + F \xi \delta \in \partial \cB] = 0$ by applying the union bound. Hence, using the Portmanteau lemma (see, e.g., Chapter 2 of \citet{vandervaart2000bk}), 
\[
	\lim_{n\to\infty}  \bP_\delta\left[ 
		\left(
			\bigcap^{L}_{l=1}  \cF^{(m_l)}_n(\og_l, \cH_1)
		\right)
	\right] 
	=
	\lim_{n\to\infty} \bP_\delta [ F\widehat S_n \in \mathfrak{B}]
	=
	\bP\left[ 
		\left(
			\bigcap^{L}_{l=1}  \cF^{(m_l)}(\og_l, \cH_1, \delta)
		\right)
	\right].
\]
The proof is complete by defining $V_{\mathrm{same}}(g, \delta)$ and $V_{\mathrm{diff}}(g, \delta)$ as in the statement of the lemma. 
\end{proof}

\begin{proof}[\textit{\textbf{Proof of Theorem \ref{prop:3.3.4}}}] 
To begin with, define $R_n \equiv \ind[T_n \neq T_n(g) \text{ for any }g \in \bG_{U, -1}]$. Following the definition in \eqref{eq:3.2}, write the local asymptotic power as the following two parts
\[
	\pi(\delta, \alpha)
	= \lim_{n\to\infty} \bP_\delta[T_n > \cv]
	= P_1 + P_2,
\]
where
\begin{align}
	P_1 & \equiv \lim_{n\to\infty} \bP_\delta[T_n > \cv, R_n = 1] ,
	\label{eq:proof3.5-p1-def-1}	\\
	P_2 & \equiv \lim_{n\to\infty} \bP_\delta[T_n > \cv, R_n = 0].		
	\label{eq:proof3.5-p1-def-2}
\end{align}
Here, $P_1$ does not allow for ``ties'' in $\{T_n > \cv\}$ whereas $P_2$ allows for ``ties.'' The goal is to derive an expression for $P_1$ and show that $P_2 = 0$.

\noindent \underline{Part 1: Computation of $P_1$} \par 
Recall that $\{T_n > \cv\}$ is the same as the event that $T_n$ is one of the largest $K \equiv \lfloor \alpha |\bG_U|\rfloor$ terms in $\{T_n(g) : g \in \bG_U\}$. Let $\bH_k$ be the collection of all distinct size $k$ subsets of $\bG_{U, -1}$ as in the statement of this proposition. Hence, the probability in \eqref{eq:proof3.5-p1-def-1} can be written as
\begin{equation}
		\bP_\delta[T_n > \cv, R_n = 1]
		= \sum^K_{k=1} P_{n,1,k},
\end{equation}
where $P_{n,1,k}$ is the probability that $T_n$ is the $k$th largest term in $\{T_n(g) : g \in \bG_U\}$, and is defined as
\begin{equation}
	\label{eq:proof3.5-0union}
	P_{n,1,k} \equiv \bP_\delta\left[ \bigcup_{\cH \in \bH_{k-1}} \cE_n(\cH)\right] ,
\end{equation}
where
\[
	\cE_n(\cH) \equiv \{ \text{$T_n > T_n(h_1)$ for all $h_1 \in \bG_{U, -1} \backslash \cH$ and $T_n(h_2) > T_n$ for all $h_2 \in  \cH$}\},
\]
with the convention that $\cE_n(\cH) = \{ \text{$T_n > T_n(h_1)$ for all $h_1 \in \bG_{U, -1}$}\}$ when $k = 1$. The strict inequalities in $\cE_n(\cH)$ follow from \eqref{eq:proof3.5-p1-def-1} that $T_n \neq T_n(g)$ for any $g \in \bG_{U, -1}$ and the definition of $\cv$. \par 

For any $\cH, \widetilde{\cH} \in \bH_{k-1}$ with $\cH \neq \widetilde{\cH}$, it must be that 
\begin{equation}
	\label{eq:intersection-empty}
	\cE_n(\cH)  \cap \cE_n(\widetilde{\cH}) = \emptyset.
\end{equation}
This is because there exists at least one $\widetilde h \in \widetilde{\cH}$ such that $T_n(\widetilde h) > T_n$ in the event $\cE_n(\widetilde{\cH})$ but $T_n > T_n(\widetilde h)$ in the event $\cE_n(\cH)$. Otherwise, $\cH = \widetilde \cH$. Since $K$ and $|\bH_{k-1}|$ are finite, it follows that
\begin{align}
	\label{eq:proof3.5-0a}
	P_1 = \sum_{k = 1}^K \sum_{\cH \in \bH_{k-1} }
	\lim_{n\to\infty} \bP_\delta[\cE_n(\cH)],
\end{align}
by the inclusion-exclusion principle and \eqref{eq:intersection-empty}.

Using Lemma \ref{lem:3.3} and \eqref{eq:proof3.5-0a}, it follows that
\[
	P_1 
	=
	\sum^K_{k=1} 
	\sum_{\cH \in \bH_{k-1}}
	\sum_{m \in \bM} 
	 \bP\left[ 
		\left(
			\bigcap^{L}_{l=1}  \cF^{(m_l)}(\og_l, \cH)
		\right)
	\right],
\]
using the same notations from Lemma \ref{lem:3.3}.

\noindent \underline{Part 2: Computation of $P_2$} \par 
Consider the case where there can be ties of $T_n$  with some $T_n(g)$ where $g \in \bG_{U, -1}$. Recall that $k$ refers to $T_n$ being the $k$th largest term. It is not possible to have such ties when $k = 1$ because this case requires $T_n > T_n(g)$ for all $g \in \bG_{U, -1}$. Thus, $P_2 = 0$ if $k = 1$.  \par 

Now, consider $k \geq 2$. The goal is also to show that the limiting probability in this case is 0. For each $k = 2, \ldots, K$, let $\bO(\cH)$ be the set that contains all subsets of $\cH$ with size $1, \ldots, |\cH|$ (i.e., it is the power set of $\cH$ excluding the empty set). By similar reasoning as in the last part and using \eqref{eq:proof3.5-p1-def-2}, $P_2$ can be written as
\begin{equation}
	\label{eq:proof3.5-6union}
	P_2 = 
	\lim_{n \to \infty}
	\sum^K_{k=2}
	P_{n,2,k},
\end{equation}
where 
\begin{equation}
	\label{eq:pf-3.5.6}
	P_{n,2,k}
	\equiv  
	\sum_{\cH \in \bH_{k - 1}}
	\sum_{\cO \in \bO(\cH)} 
	\lim_{n \to \infty} \bP_\delta[\widetilde{\cE}_n(\cH, \cO)],	
\end{equation}
and
\begin{align*}
	\widetilde{\cE}_n(\cH, \cO)
	 \equiv 
	 \{ & \text{$T_n > T_n(h_1)$ for all $h_1 \in \bG_{U, -1} \backslash \cH$,} \\
	 & \text{$T_n(h_2) > T_n$ for all $h_2 \in  \cH \backslash \cO$,}\\
	 & \text{$T_n(h_3) = T_n$ for all $h_3 \in \cO$\}},
\end{align*}
noting that $\cO$ is a subset of $\cH$ by construction. \par 
The interpretation of $P_{n,2,k}$ in \eqref{eq:pf-3.5.6} is similar to before. It collects the probability of all events such that $T_n$ is the $k$th largest term while allowing for some ties as represented by the set $\cO$. Equation \eqref{eq:pf-3.5.6} follows because $|\bH_{k-1}|$ and $|\bO(\cH)|$ are finite, and that $\cE_{n} (\cH, \cO) \cap \cE_{n} (\widetilde{\cH}, \widetilde{\cO}) = \emptyset$ unless $\cH = \widetilde{\cH}$ and $\cO = \widetilde{\cO}$. \par 

Now, applying Lemma \ref{lem:b.1} with $h = 1_q$ and using the same derivation as for \eqref{eq:pf-3.5.0} gives
\begin{align}
	\label{eq:pf-3.5.7}
	\widetilde{\cE}_n(\cH, \cO)
	= 
	\bigcup_{m \in \bM}
		\left(
		\bigcap^{L}_{l=1} 
		\widetilde\cF^{(m_l)}_n(\og_l, \cH, \cO)
		\right) ,
\end{align}
where
\begin{align*}
	\widetilde\cF^{(m_l)}_n(\og_l, \cH, \cO)
	&\equiv \begin{cases}
		 \cA_n^{(m_l)}(\og_l) & , \ \og_l \in \bG_{U,-1} \backslash \cH	\\
		  \cB_n^{(m_l)}(\og_l) &, \ \og_l \in \cH \backslash \cO \\
		  \cC_n^{(m_l)}(\og_l) & , \ \og_l \in \cO 
	\end{cases},	\\
	\cC^{(1)}_n(g)
	& \equiv 
	\left\{ \sum_{j \in \cJ_{\mathrm{same}}(g, 1_q)} \ws_{n,j} = 0 \right\}, \\
	\cC^{(2)}_n(g)
	& \equiv 
	\left\{ \sum_{j \in \cJ_{\mathrm{same}}(g, 1_q)} \ws_{n,j} \neq 0, \sum_{j \in \cJ_{\mathrm{diff}}(g, 1_q)}  \ws_{n,j} = 0 \right\},
\end{align*}
$\cA^{(m_l)}_n(\og_l)$ and $ \cB^{(m_l)}_n(\og_l)$ are as defined in the proof of Lemma \ref{lem:3.3}. \par 

Similar to the proof of \eqref{eq:pf-3.5.3}, consider any $m, \widetilde m \in \bM$ with $m \neq \widetilde m$, there must exist at least one $v = 1, \ldots, L$ such that $m_v \neq \widetilde m_v$. Assume without loss of generality that $m_v = 1$ and $\widetilde m_v = 2$. Then, it must be the case that $\widetilde{\cF}^{(m_v)}_n(\og_v, \cH, \cO) \cap \widetilde{\cF}^{(\widetilde m_v)}_n(\og_v, \cH, \cO)  = \emptyset$. This follows because $\cA^{(1)}_n(\og_v) \cap \cA^{(2)}_n(\og_v) = \emptyset$ if $\og_v \in \bG_{U,-1} \backslash \cH$, $\cB^{(1)}_n(\og_v) \cap \cB^{(2)}_n(\og_v) = \emptyset$ if $\og_v \in \cH \backslash \cO$, and  $\cC^{(1)}_n(\og_v) \cap \cC^{(2)}_n(\og_v) = \emptyset$ if $\og_v \in \cO$. As a result,
\begin{equation}
	\label{eq:pf-3.5.7}
		\left(
		\bigcap^{L}_{l=1} 
		\widetilde{\cF}^{(m_l)}_n(\og_l, \cH, \cO)
		\right) 
	\bigcap
	\left(
		\bigcap^{L}_{l=1} 
		\widetilde{\cF}^{(\widetilde m_l)}_n(\og_l, \cH, \cO)
		\right) 
	=
	\emptyset,
\end{equation}
whenever $m \neq \widetilde m$. Using \eqref{eq:pf-3.5.6}, \eqref{eq:pf-3.5.7}, \eqref{eq:pf-3.5.8}, the inclusion-exclusion principle, and the fact that $|\bM|$ is finite, it follows that 
\begin{align}
	\lim_{n\to\infty} \bP_\delta[\widetilde{\cE}_n(\cH, \cO)] 
	& = 
	\sum_{m \in \bM} 
	\lim_{n\to\infty} 
		\bP_\delta\left[ 
		\left(
			\bigcap^{L}_{l=1}  \widetilde{\cF}^{(m_l)}_n(\og_l, \cH, \cO)
		\right)
	\right]	
	\label{eq:pf-3.5.8}
\end{align}

The next step is similar to the discussion in the proof of Lemma \ref{lem:3.3}. Fix a $\cH \in \bH_{k-1}$ and $\cO \in \bO(\cH)$. Note that for each $l = 1, \ldots, L$, the event $\widetilde\cF_n^{(m_l)}(\overline g_l, \cH, \cO)$ can be written as $\{\widetilde F_l \ws_n > 0_{2\times1}\}$ if $\overline g_l \notin \cO$, where $0_{2\times 1} \equiv (0, 0)'$,  $\ws_n \equiv (\ws_{n,1}, \ldots, \ws_{n,q})'$ and $\widetilde F_l$ is a $2 \times q$ matrix defined as follows, again by recalling that $\cO$ is a subset of $\cH$ by construction.
\begin{itemize}
\item If $\overline g_l \in \bG_{U, -1} \backslash \cH$ and $m_l = 1$, then the $(i,j)$-entry of $\widetilde F_l$ is defined as in \eqref{eq:pf:lem3.3-a}.
\item If $\overline g_l \in \bG_{U, -1} \backslash \cH$ and $m_l = 2$, then the $(i,j)$-entry of $\widetilde F_l$ is defined as in \eqref{eq:pf:lem3.3-b}.
\item If $\overline g_l \in \cH \backslash \cO$ and $m_l = 1$, then the $(i,j)$-entry of $\widetilde F_l$ is defined as in \eqref{eq:pf:lem3.3-c}.
\item If $\overline g_l \in  \cH\backslash \cO$ and $m_l = 2$, then the $(i,j)$-entry of $\widetilde F_l$ is defined as in \eqref{eq:pf:lem3.3-d}.
\end{itemize}
If $\overline g_l \in \cO$, then the event $\widetilde\cF_n^{(m_l)}(\overline g_l, \cH, \cO)$ can be written as follows.
\begin{itemize}
\item If $\overline g_l  \in \cO$ and $m_l = 1$, the event $\widetilde\cF_n^{(m_l)}(\overline g_l, \cH, \cO)$ can be written as $\{\widetilde F_l'\ws_n = 0\}$, where $\widetilde F_l$ is a vector of $q$ components. The $j$th-entry of $\widetilde F_l$ is defined as
\[
	\begin{cases}
		1	& , \ \text{if $j \in \cJ_{\mathrm{same}}(\overline g_l, 1_q)$}	\\
		0 	& , \ \text{otherwise}
	\end{cases}.
\]
\item If $\overline g_l \in \cO$ and $m_l = 2$, the event $\widetilde\cF_n^{(m_l)}(\overline g_l, \cH, \cO)$ can be written as $\{\widetilde F_{l,1}'\ws_n \neq 0, \widetilde F_{l,2}'\ws_n = 0\}$, where $\widetilde F_{l,i}$ is the $i$th row of $\widetilde F_l$. The $(i,j)$-entry of $\widetilde F_l$ is defined as
\[
	\begin{cases}
		1	& , \ \text{if ($i=1$ and $j \in \cJ_{\mathrm{same}}(\overline g_l, 1_q)$) or ($i=2$ and $j \in \cJ_{\mathrm{diff}}(\overline g_l, 1_q)$)}	\\
		0 	& , \ \text{otherwise}
	\end{cases}.
\]
\end{itemize}
Let $\widetilde F$ be the matrix that stacks $\widetilde F_1, \ldots, \widetilde F_L$ by row. Using Lemma \ref{lem:b.2} and the continuous mapping theorem, 
\begin{align*}
	\widetilde F \ws_n \Dt \widetilde F Z + \widetilde F \xi \delta,
\end{align*} 
where $Z \equiv (Z_1, \ldots, Z_q)'$ and $\xi \equiv (\xi_1, \ldots, \xi_q)'$ such that $Z_j \equiv c'S_j \sim N(0, \sigma_j^2)$ and $\sigma_j^2 \equiv c'\Sigma_j c$ for all $j \in \cJ$ as before. In addition, $Z_j \indep Z_k$ for any $j \neq k$. 

Recall that for each $l = 1, \ldots, L$, $\bP[\sum_{j \in \cJ_{\text{same}}(\og_l, 1_q)} Z_j = - \sum_{j \in \cJ_{\text{same}}(\og_l, 1_q)} \xi_j \delta] = 0$ and $\bP[\sum_{j \in \cJ_{\text{diff}}(\og_l, 1_q)} Z_j = - \sum_{j \in \cJ_{\text{diff}}(\og_l, 1_q)} \xi_j \delta] = 0$ by Assumption \ref{assu:3.3.4} and applying the arguments used in the proof of Lemma \ref{lem:3.3}. Therefore, using the same reasoning as before and applying the Portmanteau lemma (see, e.g., Chapter 2 of \citet{vandervaart2000bk}), it follows that
\begin{equation}
	\label{eq:pf-3.5.8b}
	\lim_{n\to\infty}  \bP_\delta\left[ 
		\left(
			\bigcap^{L}_{l=1}  \widetilde{\cF}^{(m_l)}_n(\og_l, \cH, \cO)
		\right)
	\right] 
	=
	\bP\left[ 
		\left(
			\bigcap^{L}_{l=1}  \widetilde{\cF}^{(m_l)}(\og_l, \cH, \cO,\delta)
		\right)
	\right],
\end{equation}
where
\[
	\widetilde\cF^{(m_l)}(\og_l, \cH, \cO, \delta)
	\equiv \begin{cases}
		 \cA^{(m_l)}(\og_l, \delta) & , \ \og_l \in \bG_{U,-1} \backslash \cH	\\
		 \cB^{(m_l)}(\og_l, \delta) &, \ \og_l \in \cH \backslash \cO \\
		 \cC^{(m_l)}(\og_l,\delta) & , \ \og_l \in \cO 
	\end{cases},
\]
for any $l = 1, \ldots, L$. In the above, $\cA^{(m_l)}(\og_l, \delta)$ and $ \cB^{(m_l)}(\og_l, \delta)$ are as defined in the proof of Lemma \ref{lem:3.3}, and
\begin{align*}
	\cC^{(1)}(g, \delta)
	& \equiv 
	\left\{ 
		\sum_{j \in \cJ_{\mathrm{same}}(g, 1_q)} Z_j 
			=  -\sum_{j \in \cJ_{\mathrm{same}}(g, 1_q)} \xi_j \delta,
	\right\}, \\
	\cC^{(2)}(g, \delta)
	& \equiv 
	\left\{ 
		\sum_{j \in \cJ_{\mathrm{same}}(g, 1_q)} Z_j 
			\neq  -\sum_{j \in \cJ_{\mathrm{same}}(g, 1_q)} \xi_j \delta, 
		\sum_{j \in \cJ_{\mathrm{diff}}(g, 1_q)} Z_j 
			= -\sum_{j \in \cJ_{\mathrm{diff}}(g, 1_q)} \xi_j \delta 
	\right\},
\end{align*}
for any $g \in \bG_{U, -1}$. \par 

In the event $\cC^{(1)}(\overline g_l, \delta)$ for each $l = 1, \ldots, L$, $\cJ_{\text{same}}(\overline g_l, 1_q)$ is nonempty. As before, in terms of the notations in Assumption \ref{assu:3.3.4}, this means $w_0 = \sum_{j \in \cJ_{\mathrm{same}}(g, 1_q)} \xi_j \delta$, $w_j = 1$ for $j \in \cJ_{\text{same}}(\overline g_l, 1_q)$ and $w_j = 0$ for $j \notin \cJ_{\text{same}}(\overline g_l, 1_q)$. The event $\cC^{(1)}(\overline g_l, \delta)$ occurs with probability 0 by Assumption \ref{assu:3.3.4}. Similarly, the event $\cC^{(2)}(\overline g_l, \delta)$ occurs with probability 0 by  Assumption \ref{assu:3.3.4}. It follows that the probability in \eqref{eq:pf-3.5.8b} equals 0. 
\par 
The above argument works with any $\cH \in \bH_{k-1}$ and $\cO \in \bO(\cH)$. This implies that $P_2 = 0$ by \eqref{eq:pf-3.5.6}.  

\noindent \underline{Conclusion} \par 
Combining the results in the two parts, it follows that $\pi(\delta, \alpha) = P_1$.
\end{proof}

\begin{proof}[\textit{\textbf{Proof of Corollary \ref{cor:3.3.5}}}] For any positive integer $q$, $\alpha  \in [ \frac{1}{2^{q-1}} , \frac{1}{2^{q-2}})$ corresponds to the case where $K = \lfloor \alpha |\bG_U| \rfloor = 1$. Using Theorem \ref{prop:3.3.4}, the rejection probability is
\begin{equation}
	\label{eq:pf:3.3.5.1}
	\lim_{n\to\infty} 
	\bP_\delta[T_n > \widehat{\mathrm{cv}}_n(1-\alpha)]
	= 
	\sum_{m \in \bM}
	\lim_{n\to\infty} \bP_\delta \left[ 
		\bigcap^{L}_{l=1}
		\cF^{(m_l)}_n(\og_l)
	\right],
\end{equation}
where 
\begin{align*}
	\cF^{(1)}_n(g)
	& \equiv 
	\left\{ \sum_{j \in \cJ_{\mathrm{same}}(g, 1_q)} \ws_{n,j} > 0, \sum_{j \in \cJ_{\mathrm{diff}}(g, 1_q)}  \ws_{n,j} > 0 \right\}, \\
	\cF^{(2)}_n(g)
	& \equiv 
	\left\{ \sum_{j \in \cJ_{\mathrm{same}}(g, 1_q)} \ws_{n,j} < 0, \sum_{j \in \cJ_{\mathrm{diff}}(g, 1_q)}  \ws_{n,j} < 0 \right\},
\end{align*}
for any $g \in \bG_{U, -1}$. Note that $\cF^{(1)}_n(g) \cap \cF^{(2)}_n(h) = \emptyset$ for any $g, h \in \bG_{U, -1}$ with $g \neq h$. This is because for any $g \in \bG_{U, -1}$, the event $\cF^{(1)}_n(g)$ implies
\[
	\sum_{j \in \cJ} \ws_{n,j}
	= \sum_{j \in \cJ_{\mathrm{same}}(g, 1_q)} \ws_{n,j} +
	\sum_{j \in \cJ_{\mathrm{diff}}(g, 1_q)} \ws_{n,j} 
	> 0,
\]
and the event $\cF^{(2)}_n(g)$ implies
\[
	\sum_{j \in \cJ} \ws_{n,j}
	= \sum_{j \in \cJ_{\mathrm{same}}(g, 1_q)} \ws_{n,j} +
	\sum_{j \in \cJ_{\mathrm{diff}}(g, 1_q)} \ws_{n,j} 
	< 0.
\]
If $\{m_l\}^L_{l=1}$ are not all the same, it follows that
\[
	\bP_\delta\left[\bigcap^{L}_{l=1}
		\cF^{(m_l)}_n(\og_l)
	\right]
	= 0.
\]
Therefore, 
\begin{equation}
	\label{eq:pf:3.3.5.2}
	\lim_{n\to\infty} 
	\bP_\delta[T_n > \widehat{\mathrm{cv}}_n(1-\alpha)]
	= 
	\lim_{n\to\infty} \bP_\delta \left[ 
		\bigcap^{L}_{l=1}
		\cF^{(1)}_n(\og_l)
	\right]
	+ 
	\lim_{n\to\infty} \bP_\delta \left[ 
		\bigcap^{L}_{l=1}
		\cF^{(2)}_n(\og_l)
	\right].
\end{equation}
Note that the intersections from $l = 1, \ldots, L$ considers all the sign changes in $\bG_{U, -1}$. By expanding the intersection, 
\begin{align*}
	\bigcap^{L}_{l=1} \cF^{(1)}_n(\og_l)	
	& =
	\bigcap_{j_1 \in \cJ} \{\ws_{n,j_1} > 0\}
	\bigcap_{\{j_1, j_2 \} \subset \cJ}  \{\ws_{n,j_1} + \ws_{n,j_2} > 0\}	\\
	& \qquad \cdots
	\bigcap_{\{j_1, j_2,\ldots, j_{q-1} \} \subset \cJ}  \{\ws_{n,j_1} + \ws_{n,j_2} + \cdots + \ws_{n, j_{q-1}} > 0\} \\
	& =
	\bigcap_{j_1 \in \cJ} \{\ws_{n,j_1} > 0\},
\end{align*}
because the partial sums being positive is implied by the individual terms being positive. For the same reason, it follows that
\begin{align*}
	 \bigcap^{L}_{l=1} \cF^{(2)}_n(\og_l)	
	& =
	\bigcap_{j_1 \in \cJ} \{\ws_{n,j_1} < 0\}
	\bigcap_{\{j_1, j_2 \} \subset \cJ}  \{\ws_{n,j_1} + \ws_{n,j_2} < 0\}	\\
	& \qquad \cdots
	\bigcap_{\{j_1, j_2, \ldots, j_{q-1}\} \subset \cJ}  \{\ws_{n,j_1} + \ws_{n,j_2} + \cdots + \ws_{n, j_{q-1}} < 0\} \\
	& =
	\bigcap_{j_1 \in \cJ} \{\ws_{n,j_1} < 0\}.
\end{align*}
Therefore,
\begin{equation}
	\label{eq:pf:3.3.5.3}
	\bP_\delta \left[ 
		\bigcap^{L}_{l=1}
		\cF^{(1)}_n(\og_l)
	\right]
	= \bP_\delta \left[
		\bigcap^q_{j=1} \{ \ws_{n,j} > 0 \}
	\right]
\end{equation}
and
\begin{equation}
	\label{eq:pf:3.3.5.4}
	\bP_\delta \left[ 
		\bigcap^{L}_{l=1}
		\cF^{(2)}_n(\og_l)
	\right]
	= \bP_\delta \left[
		\bigcap^q_{j=1} \{ \ws_{n,j} < 0 \}
	\right].
\end{equation}
\par 

Suppose that Assumption \ref{assu:2.2} holds and consider the local alternative that $c'\beta = \lambda + \frac{\delta}{\sqrt{n}}$. Using Lemma \ref{lem:b.2}, applying \eqref{eq:pf:3.3.5.1} to \eqref{eq:pf:3.3.5.4}, together with the assumption that the clusters are independent, it follows that
\begin{align*}
	\lim_{n\to\infty} \bP_\delta[T_n > \widehat{\mathrm{cv}}_n(1-\alpha)]
	& = \prod^q_{j=1} \bP[ Z_j + \xi_j \delta < 0]  + \prod^q_{j=1} \bP[ Z_j + \xi_j \delta > 0]  \\
	& = 
		 \prod_{j \in \cJ} \Phi \left( -\frac{\xi_j \delta}{\sigma_j} \right)
		+
	 	 \prod_{j \in \cJ} \left[1 - \Phi \left( -\frac{\xi_j \delta}{\sigma_j} \right)\right],
\end{align*}
where $\Phi(\cdot)$ is the standard normal cdf, and by noting that $\bP[Z_j + \xi_j \delta = 0] = 0$ for any $j = 1, \ldots, q$.
\end{proof}

\begin{proof}[\textit{\textbf{Proof of Property \ref{ppt:3.7}}}] Let $\phi(\cdot)$ be the standard normal pdf. The proofs of the properties are as follows.
\begin{enumerate}
	\item We have
	\[
		\frac{\partial \pi_{\text{L}}(\delta)}{\partial \delta}
		= -\sum_{j \in \cJ} 
			\frac{\xi_j}{\sigma_j} \phi \left(-\frac{\xi_j \delta}{\sigma_j}\right) 
			\prod_{l \neq j} \Phi \left(-\frac{\xi_l \delta}{\sigma_l}\right) ,
	\]
	and
	\[
		\frac{\partial \pi_{\text{R}}(\delta)}{\partial \delta}
		= \sum_{j \in \cJ} 
		 \frac{\xi_j}{\sigma_j} \phi \left(-\frac{\xi_j \delta}{\sigma_j}\right) 
		   \prod_{l \neq j}
		   \left[1 - \Phi \left( -\frac{\xi_l \delta}{\sigma_l} \right)\right].
	\]
	The results follow because $\xi_j$, $\sigma_j$, $\phi \left(-\frac{\xi_j \delta}{\sigma_j}\right)$, and $\Phi \left(-\frac{\xi_j \delta}{\sigma_j}\right)$ are nonnegative for all $j \in \cJ$.
	\item The result follows because
	\[
		\pi_{\text{L}}(0) = \prod_{j \in \cJ} \Phi (0) = \frac{1}{2^q}
	\]
	and
	\[
		\pi_{\text{R}}(0) = \prod_{j \in \cJ} [1 - \Phi (0)] = \frac{1}{2^q}.
	\]
	\item Note that $\xi_j > 0$ and $\sigma_j > 0$ for any $j \in \cJ$. If $\delta < 0$, then
	\[
		\Phi \left( -\frac{\xi_j \delta}{\sigma_j} \right) > \Phi(0) >  1 - \Phi \left( -\frac{\xi_j \delta}{\sigma_j} \right)
	\]
	for any $j \in \cJ$. Hence, $\pi_{\mathrm{L}}(\delta) > \frac{1}{2^q} > \pi_{\mathrm{R}}(\delta)$ by multiplying the terms of $j \in \cJ$. \par 
	 If $\delta > 0$, then
	\[
		\Phi \left( -\frac{\xi_j \delta}{\sigma_j} \right) < 
		\Phi(0)
		<
		1 - \Phi \left( -\frac{\xi_j \delta}{\sigma_j} \right)
	\]
	 for any $j \in \cJ$. Hence, the result follows from multiplying the terms over $j \in \cJ$. Similarly, $\pi_{\mathrm{R}}(\delta) > \frac{1}{2^q} > \pi_{\mathrm{L}}(\delta)$ under this case.
\end{enumerate}

\end{proof}

\begin{proof}[\textit{\textbf{Proof of Proposition \ref{prop:4.6b}}}]
\text{ }
\begin{enumerate}
\item To begin with, I show Assumption \ref{assu:4.4}.\ref{assu:4.4.1} holds under Assumption \ref{assu:4.5}. Since $\widetilde\Omega = \{\omega^\star\}$, the statement I need to show is 
\begin{equation}
	\label{eq:pf:4.6.1}
	\lim_{n\to\infty} \bP[\widehat\omega_n = \omega^\star] = 1.
\end{equation}
Here, $\alpha$ and $\delta$ are fixed throughout the proof, so I write $\widehat\pi_n(\omega) \equiv \widehat\pi_n(\omega, \delta, \alpha)$ for notational simplicity. Let $\widehat\Delta_n(\omega) \equiv \widehat\pi_n(\omega) - \pi(\omega)$ for each $\omega \in \Omega$. By Assumption \ref{assu:4.5}.\ref{assu:4.5.1} and the continuous mapping theorem, $\widehat\pi_n(\omega) - \widehat\pi_n(\omega^\star) \Pt \pi(\omega) - \pi(\omega^\star)$ for any $\omega \in \Omega \backslash \{\omega^\star\}$. For a given $\omega \in \Omega \backslash \{\omega^\star\}$, this means for any $\epsilon(\omega) > 0$,
\begin{equation}
	\label{eq:pf:4.6.2}
	\lim_{n\to\infty}
	\bP[|\widehat\Delta_n(\omega) - \widehat\Delta_n(\omega^\star)| > \epsilon(\omega)]
	= 0.
\end{equation} \par 

Next, by Assumptions \ref{assu:4.5}.\ref{assu:4.5.2} and \ref{assu:4.5}.\ref{assu:4.5.3}, since $\omega^\star$ is the unique solution to \eqref{eq:4.1.1}, this means there exists $\eta > 0$ such that
\begin{equation}
	\label{eq:pf:4.6.3}
	\pi(\omega^\star) - \pi(\omega) > \eta
\end{equation}
for all $\omega \in \Omega \backslash \{\omega^\star\}$. \par 

Note that
\begin{align}
	\bP[\widehat\omega_n = \omega^\star]
	& \geq 
	\bP\left[
		\bigcap_{\omega \in \Omega \backslash \{\omega^\star\}}
		\{
			\widehat\pi_n(\omega^\star) > \widehat\pi_n(\omega)
		\}
	\right]		\notag		\\
	& = \bP\left[
		\bigcap_{\omega \in \Omega \backslash \{\omega^\star\}}
		\{
		\pi(\omega^\star) + \widehat\Delta_n(\omega^\star)
		>
		\pi(\omega) + \widehat\Delta_n(\omega)
		\}
	\right]		\notag		\\
	& = \bP\left[
		\bigcap_{\omega \in \Omega \backslash \{\omega^\star\}}
		\{
		\pi(\omega^\star) - \pi(\omega)
		>
		\widehat\Delta_n(\omega) -  \widehat\Delta_n(\omega^\star)
		\}
	\right]		\notag		\\
	& = 1 - \bP\left[
		\bigcup_{\omega \in \Omega \backslash \{\omega^\star\}}
		\{
		\pi(\omega^\star) - \pi(\omega)
		\leq
		\widehat\Delta_n(\omega) -  \widehat\Delta_n(\omega^\star)
		\}
	\right]		\notag		\\
	&  \geq 1 - \sum_{\omega \in \Omega \backslash \{\omega^\star\}} 
		\bP\left[
		\pi(\omega^\star) - \pi(\omega)
		\leq
		\widehat\Delta_n(\omega) -  \widehat\Delta_n(\omega^\star)
	\right]		\notag		\\
	&  \geq 1 - \sum_{\omega \in \Omega \backslash \{\omega^\star\}} 
		\bP\left[
		\pi(\omega^\star) - \pi(\omega)
		\leq
		| \widehat\Delta_n(\omega) -  \widehat\Delta_n(\omega^\star) | 
	\right]		\notag		\\
	& \geq  1 - \sum_{\omega \in \Omega \backslash \{\omega^\star\}} 
		\bP\left[
		\eta
		\leq
		| \widehat\Delta_n(\omega) -  \widehat\Delta_n(\omega^\star) | 
	\right],	\label{eq:pf:4.6.4}
\end{align}
where the first line follows from noting that $\pi_n(\omega^\star) > \widehat\pi_n(\omega)$ for any $\omega \in \Omega \backslash \{\omega^\star\}$ implies $\widehat\omega_n = \omega^\star$, the second line follows from the definition of $\widehat\Delta_n(\omega)$, the third line follows from rearranging terms, the fourth line follows from probability rules, the fifth line follows from the union bound, the sixth line follows from noting that the events in the fifth line imply the events in the sixth line, and the last line follows from \eqref{eq:pf:4.6.3}. \par 

Recall that $|\Omega| < \infty$. By choosing $\epsilon(\omega) = \frac{\eta}{2}$ in \eqref{eq:pf:4.6.2} for each $\omega \in \Omega \backslash \{\omega^\star\}$, it follows that
\begin{equation}
	\label{eq:pf:4.6.4}
	\lim_{n\to\infty} \bP[\widehat\omega_n = \omega^\star]
	\geq 
	1 - 
	\sum_{\omega \in \Omega \backslash \{\omega^\star\}} 
	\lim_{n\to\infty}
		\bP\left[
		\eta
		\leq
		| \widehat\Delta_n(\omega) -  \widehat\Delta_n(\omega^\star) | 
	\right]
	= 1.
\end{equation}
Therefore, \eqref{eq:pf:4.6.1} follows from the squeeze theorem and \eqref{eq:pf:4.6.4}.

\item For Assumption \ref{assu:4.4}.\ref{assu:4.4.2}, it is equivalent to
\begin{align}
	\label{eq:pf:4.6b-1}
	\lim_{n\to\infty} | 
		\bP[\phi_n(\widehat\omega_n) = 1 | \widehat\omega_n = \omega^\star] - \bP[\phi_n( \omega^\star) = 1]
	| = 0,
\end{align}
since $\widetilde\Omega = \{\omega^\star\}$ by Assumption \ref{assu:4.5}.\ref{assu:4.5.2}. \par 
Note that 
\begin{align}
	& \hspace{-25pt} \bP[\phi_n(\widehat\omega_n) = 1 | \widehat\omega_n =  \omega^\star] 
	- \bP[\phi_n( \omega^\star) = 1]		\notag	\\
	& = \frac{\bP[\phi_n(\widehat\omega_n) = 1, \widehat\omega_n = \omega^\star]}{\bP[\widehat\omega_n = \omega^\star]} 
	- \bP[\phi_n(\omega^\star) = 1]	\notag 	\\
	& = 
			\frac{\bP[\phi_n(\omega^\star) = 1, \widehat\omega_n = \omega^\star]}{\bP[\widehat\omega_n = \omega^\star]} - \bP[\phi_n(\omega^\star) = 1]		\notag 	\\
			& = 
			\frac{
			\bP[\phi_n(\omega^\star)=1] 
			- \bP[\phi_n(\omega^\star) = 1, \widehat\omega_n \neq \omega^\star]}{\bP[\widehat\omega_n = \omega^\star]}
			- \bP[\phi_n(\omega^\star) = 1]		\notag	\\
			& = 
			\bP[\phi_n(\omega^\star) = 1]
			\left(
			\frac{1}{\bP[\widehat\omega_n = \omega^\star]} -1
			\right)
			-
			\frac{
			\bP[\phi_n(\omega^\star) = 1, \widehat\omega_n \neq \omega^\star]}{\bP[\widehat\omega_n = \omega^\star]}.	\label{eq:pf:4.6b-2}
\end{align}

Since $\bP[\phi_n(\omega^\star) = 1, \widehat\omega_n \neq \omega^\star]
	\leq 
	\bP[\widehat\omega_n \neq \omega^\star]$,
it follows from \eqref{eq:pf:4.6.1} that
\[
	\lim_{n\to\infty}
	\bP[\phi_n(\omega^\star) = 1, \widehat\omega_n \neq \omega^\star]
	\leq
	1 - 
	\lim_{n\to\infty}
	\bP[\widehat\omega_n = \omega^\star]
	= 0.
\]
Hence,
\begin{equation}
		\label{eq:pf:4.6b-3}
		\lim_{n\to\infty}
	\bP[\phi_n(\omega^\star) = 1, \widehat\omega_n \neq \omega^\star]
	= 0
\end{equation}
by the squeeze theorem. Since CRS controls size of any $\omega \in \Omega$ under Assumption \ref{assu:2.2}, it must also hold for $\omega^\star$, i.e.,
\begin{equation}
	\label{eq:pf:4.6b-4}
	\lim_{n\to\infty} 
		\bP[\phi_n(\omega^\star) = 1]
		\leq \alpha.
\end{equation}

Using  \eqref{eq:pf:4.6.1}, \eqref{eq:pf:4.6b-2} to \eqref{eq:pf:4.6b-4}, it follows that
\[
	\lim_{n\to\infty}
	\left(
		\bP[\phi_n(\widehat\omega_n) = 1 | \widehat\omega_n = \omega] 
		- \bP[\phi_n(\omega) = 1]
	\right)
	= 0.
\]
This implies that Assumption \ref{assu:4.4}.\ref{assu:4.4.2} holds.
\end{enumerate}

\end{proof}

\begin{proof}[\textit{\textbf{Proof of Theorem \ref{prop:4.7}}}]
Under Assumption \ref{assu:2.2} with $\cJ$ replaced by any $\omega \in \Omega$ and the null hypothesis, the following holds
\begin{equation}
	\label{eq:pf:4.7.1}
	\lim_{n\to\infty} \bP[\phi_n(\omega) = 1] \leq \alpha
	\qquad \text{for each $\omega \in \Omega$.}
\end{equation}
This is because for any given grouping of clusters $\omega \in \Omega$, CRS controls size.  \par 

Next, consider the test based on the data-driven procedure under the null hypothesis, where $\widehat\omega_n$ is the grouping of clusters given by the procedure:
\begin{align*}
	& \hspace{-10pt} \lim_{n\to\infty} \bP[\phi_n(\widehat\omega_n) = 1]		\\
	& = \lim_{n\to\infty} \sum_{\omega \in \Omega} \bP[\phi_n(\widehat\omega_n) = 1|\widehat\omega_n = \omega] \bP[\widehat\omega_n = \omega] \\
	& = \lim_{n\to\infty} \sum_{\omega \in \Omega} \bP[\phi_n(\omega) = 1|\widehat\omega_n = \omega] \bP[\widehat\omega_n = \omega] \\
	& = \lim_{n\to\infty} \sum_{\omega \in \Omega} \bP[\phi_n(\omega) = 1|\widehat\omega_n = \omega] \left(\bP[\widehat\omega_n = \omega, \widehat\omega_n \in \widetilde\Omega] + \bP[\widehat\omega_n = \omega, \widehat\omega_n \notin \widetilde\Omega] \right)\\
	& = \sum_{\omega \in \Omega}  \lim_{n\to\infty} |\bP[\phi_n(\omega) = 1|\widehat\omega_n = \omega] - \bP[\phi_n(\omega) = 1] + \bP[\phi_n(\omega) = 1] | \\
	& \qquad\qquad \cdot  \lim_{n\to\infty} \left(\bP[\widehat\omega_n = \omega, \widehat\omega_n \in \widetilde\Omega] + \bP[\widehat\omega_n = \omega, \widehat\omega_n \notin \widetilde\Omega] \right) \\
	& = \sum_{\omega \in \Omega}  \lim_{n\to\infty} |\bP[\phi_n(\omega) = 1|\widehat\omega_n = \omega] - \bP[\phi_n(\omega) = 1] + \bP[\phi_n(\omega) = 1] | \\
	& \qquad\qquad \cdot  \lim_{n\to\infty} \bP[\widehat\omega_n = \omega, \widehat\omega_n \in \widetilde \Omega] \\
	& \leq \sum_{\omega \in \Omega} \lim_{n\to\infty} |\bP[\phi_n(\omega) = 1] | \cdot  \lim_{n\to\infty} \bP[\widehat\omega_n = \omega, \widehat\omega_n \in \widetilde \Omega] \\
	& \leq \alpha \sum_{\omega \in \Omega} \lim_{n\to\infty} \bP[\widehat\omega_n = \omega, \widehat\omega_n \in \widetilde \Omega] \\
	& = \alpha \lim_{n\to\infty} \sum_{\omega \in \Omega} \bP[\widehat\omega_n = \omega, \widehat\omega_n \in \widetilde \Omega] \\
	& \leq \alpha,
\end{align*}
where the first equality uses the law of total probability, the second equality applies the condition that $\widehat\omega_n = \omega$, the third equality uses the law of total probability, the fourth equality uses the fact that $|\Omega|$ is finite and adds and subtracts, the fifth equality uses Assumption \ref{assu:4.4}.\ref{assu:4.4.1} on the existence of the set $\widetilde\Omega$ such that $\lim_{n\to\infty} \bP[\widehat\omega_n \notin \widetilde \Omega] = 0$, the first inequality uses Assumption \ref{assu:4.4}.\ref{assu:4.4.2}, the triangle inequality, the next inequality uses \eqref{eq:pf:4.7.1}, the next equality exchanges the summation and limit using the fact that $|\Omega|$ is finite again, and the last equality uses $\{\widehat\omega_n = \omega, \widehat\omega_n \in \widetilde\Omega\} \subseteq \{\widehat\omega_n = \omega\}$ and that probabilities over $\omega \in \Omega$ sum to 1. \end{proof}

\begin{proof}[\textit{\textbf{Proof of Proposition \ref{prop:4.9}}}]
Assume without loss of generality that $\delta > 0$. Hence, Algorithm \ref{algo:4.2.1} uses optimization problem \eqref{eq:4.2.6b} to approximate the solution of problem \eqref{eq:4.2.3}. Following the statement of the proposition, let $\pi^\star$ be the optimal value to optimization problem \eqref{eq:4.2.3}. In addition, let $\{z^\star_{j,r}\}$ be the solution to \eqref{eq:4.2.3} such that it gives the objective value $\pi^\star$. Denote
\begin{equation}
	\label{eq:4.9.0a}
	\pi^\star_{\text{L}}
	\equiv
	\sum^{\overline{q}}_{j=1} \sum^{\overline{q}}_{r=1}
	z_{j,r}^\star \log \Psi_{j,r}
\end{equation}
and
\begin{equation}
	\label{eq:4.9.0b}
	\pi^\star_{\text{R}}
	\equiv
	\sum^{\overline{q}}_{j=1} \sum^{\overline{q}}_{r=1}
	z_{j,r}^\star \log (1 - \Psi_{j,r}),
\end{equation}
so that $\pi^\star = \pi^\star_{\text{L}} + \pi^\star_{\text{R}}$. Moreover, let 
\[
	\cL \equiv 
	\left\{\sum^{\overline{q}}_{j=1} \sum^{\overline{q}}_{r=1} z_{j,r} \log \Psi_{j,r} : \text{$\{z_{j,r}\}$ is a feasible solution to \eqref{eq:4.2.3}}
	\right\}
\]
be the set that collects all unique values of $\sum^{\overline{q}}_{j=1} \sum^{\overline{q}}_{r=1} z_{j,r} \log \Psi_{j,r}$ based on the feasible solutions to \eqref{eq:4.2.3}. Since $\overline q$ is fixed, it follows that $|\cL|$ is finite. Let $\pi_{\text{L}}^{(1)} < \pi_{\text{L}}^{(2)} < \cdots < \pi_{\text{L}}^{(|\cL|)}$ be the ordered and distinct values of in $\cL$. Here, $\pi^{(1)}_{\text{L}} > 0$ because $\delta$, as well as $\xi_{j,r}$ and $\sigma_{j,r}$ for $j, r= 1,\ldots, \overline q$ are finite. In addition, $\pi^\star_{\text{L}} \in \cL$ because $\{z^\star_{j,r}\}$ is a feasible solution to \eqref{eq:4.2.3}. \par 

Let $\eta \equiv \min_{\ell \in \{1, \ldots, |\cL| - 1\}} (\log\pi_{\text{L}}^{(\ell + 1)} - \log\pi_{\text{L}}^{(\ell)})$. Note that $\eta > 0$ by construction. In addition, let $A_0 \equiv \lceil  \frac{2(\log \pi^{(|\cL|)}_{\text{L}} - \log\pi^{(1)}_{\text{L}})}{\eta} \rceil$. 
Next, choose $\epsilon_{\text{lb}} = \pi^{(1)}_{\text{L}}$ and $\epsilon_{\text{ub}} = \pi^{(|\cL|)}_{\text{L}}$. Partition $[\epsilon_{\text{lb}},\epsilon_{\text{ub}}]$ into $A_0$ intervals, so that the length of each interval $\log \epsilon_a - \log \epsilon_{a-1}$ is at most $\frac{\eta}{2}$. For each interval, it can contain at most one $\pi^{(\ell)}_{\text{L}}$ for $\ell = 1, \ldots, |\cL|$ because the width of each interval is strictly less than $\eta$. If $\pi^\star_{\text{L}}$ lies in the boundary of an interval, i.e., there exists $a'$ such that $\epsilon_{a'} = \pi^\star_{\text{L}}$, partition the interval $[\epsilon_{\text{lb}}, \epsilon_{\text{ub}}]$ into $A_0 + 1$ equally-spaced intervals instead, and check if $\pi^\star_{\text{L}}$ lies in the interior of one of the intervals. If $\pi^\star_{\text{L}}$ is still on the boundary, repeat the above step again by partitioning $[\epsilon_{\text{lb}}, \epsilon_{\text{ub}}]$ into finer intervals, until $\pi^\star_{\text{L}}$ lies in the interior of one of the intervals. In addition, define $[\underline\epsilon, \overline\epsilon]$ as the subinterval that contains $\pi^\star_{\text{L}}$.\par 

Consider any choices of $\log \epsilon_{a'-1}$ and $\log \epsilon_{a'}$ with $\log \epsilon_{a'-1} \neq \underline\epsilon$ and $\log \epsilon_{a'} \neq \overline\epsilon$ such that \eqref{eq:4.2.6} is feasible. Let $\{\widetilde z_{j,r}\}$ be the corresponding optimal solution to \eqref{eq:4.2.6} with this choice of interval $[\log \epsilon_{a'-1}, \log \epsilon_{a'}]$. It must be the case that 
\[
\sum^{\overline{q}}_{j=1} \sum^{\overline{q}}_{r=1} \widetilde{z}_{j,r} \log \Psi_{j,r} \neq \pi^\star_{\text{L}}
\]
by the choice of $a'$. Assume to the contrary that 
\begin{equation}
	\label{eq:prop:4.9.1}
	\sum^{\overline{q}}_{j=1} \sum^{\overline{q}}_{r=1} \widetilde{z}_{j,r} \log \Psi_{j,r} + \sum^{\overline{q}}_{j=1} \sum^{\overline{q}}_{r=1} \widetilde{z}_{j,r} \log (1 - \Psi_{j,r}) > \pi^\star.
\end{equation}
Note that $\{\widetilde z_{j,r}\}$ is feasible to \eqref{eq:4.2.3} with this choice of $[\log \epsilon_{a'-1}, \log \epsilon_{a'}]$ because \eqref{eq:4.2.3} contains a subset of constraints of \eqref{eq:4.2.6}. But \eqref{eq:prop:4.9.1} contradicts that $\{z_{j,r}^\star\}$ is an optimal solution to \eqref{eq:4.2.3} because $\{\widetilde z_{j,r}\}$ gives a higher optimal value.

Let (P) represents the optimization problem \eqref{eq:4.2.6} with $\log \epsilon_{a-1} = \underline\epsilon$ and $\log \epsilon_a = \overline\epsilon$. Note that (P) is feasible since $\{z^\star_{j,r}\}$ is one such feasible solution. Let $\{\overline z_{j,r}\}$ be the optimal solution to (P). By construction, it must be the case that $\sum^{\overline{q}}_{j=1} \sum^{\overline{q}}_{r=1} \overline{z}_{j,r} \log \Psi_{j,r} = \pi^\star_{\text{L}}$. In addition, $\{\overline z_{j,r}\}$ must also be feasible to \eqref{eq:4.2.3} because \eqref{eq:4.2.3} contains a subset of the constraints of \eqref{eq:4.2.6}. \par

Now, let $\overline \pi_{\text{R}} \equiv \sum^{\overline{q}}_{j=1} \sum^{\overline{q}}_{r=1} \overline{z}_{j,r} \log (1 - \Psi_{j,r})$. It remains to show that $\overline \pi_{\text{R}} = \pi^\star_{\text{R}}$. First, suppose $\overline \pi_{\text{R}}  > \pi^\star_{\text{R}}$. This implies that $\overline \pi_{\text{L}} + \overline \pi_{\text{R}} > \pi^\star_{\text{L}} + \pi^\star_{\text{R}}$. But $\{\overline z_{j,r}\}$ is feasible to \eqref{eq:4.2.3}, this contradicts that $\{z^\star_{j,r}\}$ is an optimal solution to \eqref{eq:4.2.3}. Next, suppose $\overline \pi_{\text{R}} < \pi^\star_{\text{R}}$. This implies that $\overline \pi_{\text{L}} + \overline \pi_{\text{R}} < \pi^\star_{\text{L}} + \pi^\star_{\text{R}}$. But $\{z^\star_{j,r}\}$ is also feasible to (P) because it satisfies the constraint 
\[
	\sum^{\overline{q}}_{j=1} \sum^{\overline{q}}_{r=1} z_{j,r}^\star \log \Psi_{j,r} 
	= z_{\text{L}}^\star
	\in  [\log \underline \epsilon, \log \overline \epsilon].
\]
Thus, this contradicts that $\{\overline z_{j,r}\}$ is an optimal solution to (P). Therefore, it must be the case that $\overline \pi_{\text{R}} = \pi^\star_{\text{R}}$ and the proof is complete.

If $\delta < 0$, then the above proof can be modified by changing \eqref{eq:4.2.6b} to \eqref{eq:4.2.6} and reversing the roles of the terms related to \eqref{eq:4.9.0a} and \eqref{eq:4.9.0b} appropriately. 
\end{proof}

\section{More details on the empirical application} \label{app:d}

This section describes how the calibrated simulation exercise is conducted using the data of \citet{dinceccokatz2016ej}. I create a balanced panel based on their data. Let $T$ be the longest time period of the data. Here, I focus on column (1) of Table 3 of \citet{dinceccokatz2016ej}, so that the regression follows \eqref{eq:emp-ej-1}. The calibrated simulation procedure is as follows.
\begin{enumerate}
	\item Estimate the model to obtain $(\widehat\beta_0, \widehat\beta_1, \widehat\beta_2, \{\widehat\mu_j\}_{j\in\cJ})$ and obtain the residuals as 
	\[
		\widehat{U}_{t,j}
		= Y_{t,j} - 
		(\widehat\beta_0 + \widehat\beta_1C_{t,j} + \widehat\beta_2 L_{t,j} + \widehat\mu_j),
	\]
	for all $j \in \cJ$ and $t \in \{1, \ldots, T\} \equiv \cT$.
	\item \label{calib:step2} Use the residuals to fit an AR(1) model 
	\[
		{U}_{t,j} = \rho_j U_{t-1,j} + \epsilon_{t,j},
	\]
	where $\epsilon_{t,j} \sim N(0, \nu_j^2)$ for each country $j \in \cJ$. Hence, there is an estimate of $\widehat\rho_j$ and $\widehat\nu_j^2$ for each country $j \in \cJ$.
	\item Conduct a Monte Carlo simulation with $B$ replications. Each replication $b = 1,\ldots, B$ is as follows.
	\begin{enumerate}
		\item For each country $j \in \cJ$:
		\begin{enumerate}
			\item Draw the residuals $\widetilde{U}_{t,j}^{(b)}$ based on the model estimated in Step \ref{calib:step2} from $t = 1$ to $t = T$.
			\item Next, set the values of the treatment variables to create variation across countries and time.  
			\begin{itemize}
			\item If country $j$ has variation in $C_{t,j}$ in the original data, set $t_{C, j} = \lfloor \frac{3}{4} T \rfloor - 5j$. Otherwise, set $t_{C,j} = 0$.
			\item If country $j$ has variation in $L_{t,j}$ in the original data, set $t_{L,j} = \lfloor \frac{3}{4} T \rfloor - 8(11 - j)$. Otherwise, set $t_{L,j} = 0$.
			\end{itemize}
			\item Set $C_{t,j} = \ind[t > t_{C,j}]$ and $L_{t,j} = \ind[t > t_{L,j}]$ for all $t \in \cT$.
		\end{enumerate}
		\item Suppose the goal is to perform inference on the coefficient of $C_{t,j}$. Then,
		\begin{enumerate}
			\item Let $\widehat\beta_1 + \Delta$ be the value of the alternative. 
			\item Keep $\widehat\beta_2$ unchanged. 
			\item Generate the outcome as
		\[
			\widetilde{Y}_{t,j}^{(b)}
			= \widehat\beta_0 
			+ (\widehat\beta_1 + \Delta) C_{t,j} + \widehat\beta_2 L_{t,j} + \widehat\mu_j  + \widetilde{U}_{t,j}^{(b)},
		\]
		for all $j \in \cJ$ and $t \in \cT$.
		\end{enumerate}
		\item Suppose the goal is to perform inference on the coefficient of $L_{t,j}$. Then,
		\begin{enumerate}
			\item Let $\widehat\beta_2 + \Delta$ be the value of the alternative. 
			\item Keep $\widehat\beta_1$ unchanged.
			\item Generate the outcome as
		\[
			\widetilde{Y}_{t,j}^{(b)}
			= \widehat\beta_0 
			+ \widehat\beta_1 C_{t,j} + (\widehat\beta_2 + \Delta)  L_{t,j} + \widehat\mu_j +  \widetilde{U}_{t,j}^{(b)},
		\]
		for all $j \in \cJ$ and $t \in \cT$.
		\end{enumerate}
	\end{enumerate}
\end{enumerate}
The procedure for the other two columns is similar when time fixed effects and/or time trends are added to the regression equation.

\newpage
\addcontentsline{toc}{section}{References}
{\small{
	\singlespacing{
	\bibliographystyle{ecta}
	\bibliography{myrefs}
	}
}}

\end{document}